\documentclass[pra,twocolumn,amsmath,amssymb,a4paper]{revtex4} 
\usepackage{geometry}
\usepackage{epstopdf}
\usepackage{graphicx}
\usepackage{mathrsfs}
\usepackage{amssymb}
\usepackage{amsmath}
\usepackage{amsthm}
\usepackage{bm}
\usepackage{hyperref}

\newtheorem{proposition}{Proposition}

\newtheorem{corollary}{Corollary}

\theoremstyle{definition}

\newtheorem{example}{Example}
\newtheorem{remark}{Remark}

\newcommand{\bra}[1]{\langle #1|}
\newcommand{\ket}[1]{| #1 \rangle }
\newcommand{\ip}[2]{{\langle #1|}{ #2 \rangle }}

\begin{document}
	
	\title{Positive tensor products of maps and $n$-tensor-stable positive qubit maps}
	
	\author{Sergey N. Filippov}
	
	\affiliation{Moscow Institute of Physics and Technology,
		Institutskii Per. 9, Dolgoprudny, Moscow Region 141700, Russia}
	
	\author{Kamil Yu. Magadov}
	
	\affiliation{Moscow Institute of Physics and Technology,
		Institutskii Per. 9, Dolgoprudny, Moscow Region 141700, Russia}
	
	\begin{abstract}
		We analyze positivity of a tensor product of two linear qubit
		maps, $\Phi_1 \otimes \Phi_2$. Positivity of maps $\Phi_1$ and
		$\Phi_2$ is a necessary but not a sufficient condition for
		positivity of $\Phi_1 \otimes \Phi_2$. We find a non-trivial
		sufficient condition for positivity of the tensor product map
		beyond the cases when both $\Phi_1$ and $\Phi_2$ are completely
		positive or completely co-positive. We find necessary and
		(separately) sufficient conditions for $n$-tensor-stable positive
		qubit maps, i.e. such qubit maps $\Phi$ that $\Phi^{\otimes n}$ is
		positive. Particular cases of 2- and 3-tensor-stable positive
		qubit maps are fully characterized, and the decomposability of
		2-tensor-stable positive qubit maps is discussed. The case of
		non-unital maps is reduced to the case of appropriate unital maps.
		Finally, $n$-tensor-stable positive maps are used in
		characterization of multipartite entanglement, namely, in the
		entanglement depth detection.
	\end{abstract}
	
	\maketitle
	
	
	\section{\label{section-introduction} Introduction}
	
	Tensor product structures play a vital role in quantum information
	theory: entanglement of quantum states is defined with respect to
	a particular bipartition~\cite{werner-1989} or multipartition
	(see, e.g., the reviews~\cite{horodecki-2009,guhne-toth-2009});
	communication via quantum channels involves multiple uses of the
	same channel, which results in the map of the form $\Phi^{\otimes
		n}$ (see, e.g.,~\cite{holevo-giovannetti-2012}); propagation of
	multipartite physical signals through separated communication
	lines $\Phi_1$ and $\Phi_2$ is described by a tensor product of
	corresponding maps $\Phi_1 \otimes \Phi_2$; local operations and
	measurements have the tensor product structure too. Properties of
	quantum channels may drastically change with tensoring as it takes
	place, for instance, in superactivation of zero-error capacities
	\cite{cubitt-2011,shirokov-2015}.
	
	Positive maps, in their turn, are an important auxiliary tool in
	quantum information theory and are widely used in the analysis of
	bipartite
	entanglement~\cite{peres-1996,horodecki-1996,terhal-2001,chen-2004,breuer-2006,hall-2006,chruscinski-kossakowski-2007,ha-2011,chruscinski-sarbicki-2014,collins-2016},
	multipartite entanglement~\cite{huber-2014,lancien-2015},
	entanglement
	distillation~\cite{horodecki-1998,horodecki-1999,divincenzo-2000,clarisse-2005},
	distinguishability of bipartite
	states~\cite{yu-2014,horodecki-2015,bandyopadhyay-2015},
	description of open system
	dynamics~\cite{shaji-2005,carteret-2008,chruscinski-2014},
	monotonicity of relative entropy~\cite{muller-hermes-2015}, and
	evaluation of quantum channel
	capacities~\cite{holevo-werner-2001}.
	
	Positivity of linear maps under tensor powers was analyzed in the
	recent seminal paper~\cite{muller-hermes-2016}, where the notions
	of $n$-tensor-stable positive and tensor-stable positive maps were
	introduced. Tensor-stable positive maps were found to provide new
	bounds on quantum channel capacities.
	
	The aim of this paper is to study positivity of the tensor product
	maps $\Phi_1 \otimes \Phi_2$, where both $\Phi_1$ and $\Phi_2$ are
	qubit maps (${\cal M}_2 \mapsto {\cal M}_2$). We focus special
	attention on 2-tensor-stable positive maps $\Phi$, i.e. such maps
	$\Phi$ that $\Phi^{\otimes 2}$ is positive, and then extend our
	results to 3- and $n$-tensor-stable positive maps.
	
	The paper is organized as follows.
	
	In Sec.~\ref{section-notation}, we review notations and general
	properties of linear maps, and formulate some sufficient
	conditions for positivity of the tensor product map $\Phi_1
	\otimes \Phi_2$. In Sec.~\ref{section-depolarizing}, an exact
	characterization of bipartite locally depolarizing positive maps
	is presented. In Sec.~\ref{section-unital-general}, sufficient
	conditions for positivity of the tensor product unital map $\Phi_1
	\otimes \Phi_2$ are derived. In
	Sec.~\ref{section-2-tensor-stable-positive-unital}, we find the
	necessary and sufficient condition for 2-tensor positivity of
	unital qubit maps. Sec.~\ref{section-decomposable} is devoted to the question of decomposability of the tensor products of qubit maps. In
	Sec.~\ref{section-2-tensor-stable-positive-non-unital},
	2-tensor-stable positivity of non-unital qubit maps is studied by
	a reduction to the problem of 2-tensor-stable positivity of
	corresponding unital maps. In
	Sec.~\ref{section-3-tensor-stable-positive-unital}, criteria for
	3-tensor positivity of unital qubit maps are found and checked
	numerically. In Sec.~\ref{section-n-tsp}, we find necessary and
	(separately) sufficient conditions for $n$-tensor-stable positive
	maps. Sec.~\ref{section-witness} is devoted to witnessing
	particular forms of multipartite entanglement via
	$n$-tensor-stable positive maps. In
	Sec.~\ref{section-conclusions}, brief conclusions are given.
	
	
	\section{\label{section-notation} Notations and general properties}
	Consider a finite dimensional Hilbert space (unitary space) ${\cal
		H}_d$, ${\rm dim}{\cal H} = d$, and the set ${\cal B}({\cal H}_d)$
	of operators acting on ${\cal H}_d$. The operator $R\in{\cal
		B}({\cal H}_d)$ is called positive semidefinite if $\bra{\psi} R
	\ket{\psi} \geqslant 0$ for all vectors $\ket{\psi}\in{\cal H}_d$
	(hereafter we use the Dirac notation). For positive semidefinite
	operators $R$ we write $R \geqslant 0$. We will denote the set of
	all positive semidefinite operators by ${\cal B}({\cal H}_d)^{+}$.
	The linear map $\Phi: {\cal B}({\cal H}_d)^{+} \mapsto {\cal
		B}({\cal H}_d)^{+}$ is called positive. By ${\rm Id}_k$ denote the
	identity transformation on ${\cal B}({\cal H}_k)$. The linear map
	$\Phi$ is called $k$-positive if the map $\Phi\otimes{\rm Id}_k$
	is positive. 2-positive maps of the form $\Phi^{\otimes n}$ are
	analyzed in the paper~\cite{stormer-2010} and play an important
	role in the distillation problem~\cite{divincenzo-2000}. A linear
	map $\Phi$ is called completely positive if it is $k$-positive for
	all $k \in \mathbb{N}$ (see, e.g.,~\cite{breuer-petruccione}). By
	$\top_d$ we denote the transposition map on ${\cal B}({\cal H}_d)$
	associated with some orthonormal basis $\{\ket{i}\}_{i=1}^d$ in
	${\cal H}_d$, $X^{\top}=\sum_{i,j} \ket{i} \bra{j} X \ket{i}
	\bra{j}$. Maps of the form $\top\circ\Phi$, where $\Phi$ is
	completely positive, are called completely co-positive. Duality
	relations between cones of different maps are discussed, e.g.,
	in~\cite{skowronek-2009}.
	
	A linear map $\Phi: {\cal B}({\cal H}_d) \mapsto {\cal B}({\cal
		H}_d)$ is called $n$-tensor-stable positive if the map
	$\Phi^{\otimes n}$ is positive~\cite{muller-hermes-2016}.
	Obviously, if $m > n$, then the set of $n$-tensor-stable positive
	maps comprises the set of $m$-tensor-stable positive maps (nested
	structure). A linear map $\Phi: {\cal B}({\cal H}_d) \mapsto {\cal
		B}({\cal H}_d)$ is called tensor-stable positive (or tensor
	product positive) if it is $n$-tensor-stable positive for all $n
	\in {\mathbb N}$~\cite{muller-hermes-2016,hayashi-2006}.
	Completely positive and completely co-positive maps $\Phi$ are
	trivial tensor-stable positive maps~\cite{muller-hermes-2016}.
	
	In subsequent sections, we exploit some properties of maps with
	regard to their action on entangled states. Quantum states are
	described by density operators, i.e. positive semidefinite
	operators $\varrho \in {\cal B}({\cal H}_d)^+$ with unit trace,
	${\rm tr}\varrho = \sum_{i=1}^d \bra{i}\varrho\ket{i} =1$. A
	positive semidefinite operator $R \in ({\cal B}({\cal H}_{d_1})
	\otimes {\cal B}({\cal H}_{d_2}))^+$ is called
	separable~\cite{werner-1989} if it can be represented in the form
	$\varrho = \sum_{k} R_k^{(1)} \otimes R_k^{(2)}$, where $R_k^{(1)}
	\in {\cal B}({\cal H}_{d_1})^+$ and $R_k^{(2)} \in {\cal B}({\cal
		H}_{d_2})^+$, otherwise $R$ is called entangled. Denote the cone
	of separable operators by ${\cal S}({\cal H}_{d_1}\otimes{\cal
		H}_{d_2})$. We will refer to completely positive maps $\Phi: {\cal
		B}({\cal H}_d) \mapsto {\cal B}({\cal H}_d)$ of the form $\Phi[X]
	= \sum_j {\rm tr}[E_j X] R_j$ with $E_j,R_j \geqslant 0$ as
	entanglement breaking (quantum--classical--quantum,
	measure-and-prepare)~\cite{holevo-1998,king-2002,shor-2002,ruskai-2003,horodecki-2003,holevo-2008}.
	A positive map $\Phi: ({\cal B}({\cal H}_{d_1}) \otimes {\cal
		B}({\cal H}_{d_2}))^+ \mapsto {\cal S}({\cal H}_{d_1}\otimes{\cal
		H}_{d_2})$ is called positive entanglement
	annihilating~\cite{moravcikova-ziman-2010,filippov-ziman-2013,filippov-melnikov-ziman-2013,filippov-ziman-2014}.
	
	Action of a linear map $\Phi: {\cal B}({\cal H}_d) \mapsto {\cal
		B}({\cal H}_d)$ can be defined through the Choi operator
	$\Omega_{\Phi} \in {\cal B}({\cal H}_d) \otimes {\cal B}({\cal
		H}_d)$ via the so-called Choi-Jamio{\l}kowski isomorphism
	~\cite{pillis-1967,jamiolkowski-1972,choi-1975} reviewed
	in~\cite{jiang-2013,majewski-2013}:
	\begin{eqnarray}
	&& \label{choi-matrix} \Omega_{\Phi} = (\Phi \otimes {\rm
		Id}_d) [\ket{\psi_+}\bra{\psi_+}], \\
	&& \label{map-through-choi} \Phi [X] = d \, {\rm tr}_2 [
	\Omega_{\Phi} (I \otimes X^{\top}) ],
	\end{eqnarray}
	
	\noindent where $\ket{\psi_+} = \frac{1}{\sqrt{d}} \sum_{i=1}^{d}
	\ket{i} \otimes \ket{i}$ is a maximally entangled state, $I$ is
	the identity operator on ${\cal H}_d$, ${\rm tr}_2 [Y] =
	\sum_{i=1}^d (I \otimes \bra{i}) Y (I \otimes \ket{i})$ denotes
	the partial trace operation for operators $Y \in {\cal B}({\cal
		H}) \otimes {\cal B}({\cal H})$.
	
	Let us remind the known properties of Choi operator:
	
	\begin{enumerate}
		\item $\Phi$ is positive if and only if $\Omega_{\Phi}$ is
		block-positive, i.e. $\bra{\varphi} \otimes \bra{\chi}
		\Omega_{\Phi} \ket{\varphi} \otimes \ket{\chi} \geqslant 0$ for
		all $\ket{\varphi},\ket{\chi}\in{\cal
			H}_d$~\cite{jamiolkowski-1972};
		
		\item $\Phi$ is completely positive (quantum operation) if and
		only if $\Omega_{\Phi} \geqslant 0$~\cite{choi-1975};
		
		\item $\Phi$ is entanglement breaking if and only if
		$\Omega_{\Phi}$ is separable (see, e.g.,~\cite{horodecki-2003});
		
		\item $\Phi$ is positive entanglement annihilating if and only if
		${\rm tr}[\Omega_{\Phi} \xi_{1|2} \otimes R_{12} ] \geqslant 0$
		for all $R \in ({\cal B}({\cal H}_{d_1}) \otimes {\cal B}({\cal
			H}_{d_2}))^+$ and all block-positive operators $\xi_{1|2} \in
		{\cal B}({\cal H}_{d_1}) \otimes {\cal B}({\cal
			H}_{d_2})$~\cite{filippov-ziman-2013}.
	\end{enumerate}
	
	The general problem addressed in this paper is to determine under
	which conditions a tensor product $\Phi_1 \otimes \Phi_2$ of two
	linear maps $\Phi_1: {\cal B}({\cal H}_{d_1}) \mapsto {\cal
		B}({\cal H}_{d_1})$ and $\Phi_2: {\cal B}({\cal H}_{d_2}) \mapsto
	{\cal B}({\cal H}_{d_2})$ is a positive map. Acting on a
	factorized positive operator $R_1 \otimes R_2 \geqslant 0$, it is
	not hard to see that the positivity of maps $\Phi_1$ and $\Phi_2$
	is a necessary condition. This condition, however, is not
	sufficient in general as $({\cal B}({\cal H}_{d_1}))^+ \otimes
	({\cal B}({\cal H}_{d_2}))^+ \subsetneq ({\cal B}({\cal H}_{d_1})
	\otimes {\cal B}({\cal H}_{d_2}))^+$. (Characterization of the
	cone $({\cal B}({\cal H}_{d_1}))^+ \otimes ({\cal B}({\cal
		H}_{d_2}))^+$ is given in Ref.~\cite{majewski-2001}.) For
	instance, the maps $\Phi_1 = {\rm Id}$ and $\Phi_2 = \top$ are
	both positive, but the map $\Phi_1 \otimes \Phi_2 = {\rm Id}
	\otimes \top$ is not positive. An apparent sufficient condition
	for positivity of the map $\Phi_1 \otimes \Phi_2$ is
	$\{\Phi_1\otimes\Phi_2$ is completely positive or completely
	co-positive$\}$, which takes place if $\Phi_1$ and $\Phi_2$ are
	both completely positive, or if $\Phi_1$ and $\Phi_2$ are both
	completely co-positive.
	
	Generalization of the problem to a number of maps $\Phi_1, \ldots,
	\Phi_n$ is to determine when the map $\bigotimes_{k=1}^{n} \Phi_n$
	is positive. Setting all the maps $\Phi_i$ to be identical
	($\Phi_i = \Phi$), we get the problem of characterizing
	$n$-tensor-stable positive maps posed in
	Ref.~\cite{muller-hermes-2016}.
	
	We restrict our analysis to the case of linear qubit maps $\Phi_i:
	{\cal B}({\cal H}_2) \mapsto {\cal B}({\cal H}_2)$. It was shown
	in Ref.~\cite{muller-hermes-2016} that all tensor-stable positive
	qubit maps are trivial (completely positive or completely
	co-positive). However, $n$-tensor-stable positive qubit maps for a
	fixed $n$ are not necessarily trivial and their characterization
	is still missing, so we partially fill this gap in the present
	paper. Also, we provide a full characterization for the cases
	$n=2$ and $n=3$. First, we obtain results for unital maps, i.e.
	such linear maps $\Phi$ that $\Phi[I]=I$. Then, we extend these
	results to the case of non-unital maps.
	
	Denote the concatenation of two maps $\Phi$ and $\Lambda$ by $\Phi
	\circ \Lambda$, i.e. $(\Phi \circ \Lambda) [X] = \Phi \left[
	\Lambda[X] \right]$.
	
	\begin{proposition}
		\label{proposition-positive-one-sided} Suppose a map
		$\Phi_1\otimes\Phi_2$ is positive entanglement annihilating and
		${\cal P}$ is positive, then the maps $\Phi_1\otimes({\cal
			P}\circ\Phi_2)$ and $({\cal P}\circ\Phi_1)\otimes\Phi_2$ are
		positive.
	\end{proposition}
	\begin{proof}
		By definition of positive entanglement annihilating map, for any
		positive semidefinite operator $R$ we have:
		$(\Phi_1\otimes\Phi_2)[R] = \sum_{k} R_k^{(1)} \otimes R_k^{(2)}
		\geqslant 0$, where $R_k^{(1)} \geqslant 0$ and $R_k^{(2)}
		\geqslant 0$. Since ${\cal P}[R_k^{(2)}] \geqslant 0$, the
		operator $\left( \Phi_1\otimes({\cal P}\circ\Phi_2) \right) [R] =
		\sum_{k} R_k^{(1)} \otimes {\cal P}[R_k^{(2)}] \geqslant 0$ for
		all $R \geqslant 0$. Similarly, $\left( ({\cal
			P}\circ\Phi_1)\otimes\Phi_2 \right) [R] \geqslant 0$ for all $R \geqslant 0$.
	\end{proof}
	
	Proposition~\ref{proposition-positive-one-sided} enables one to
	use known criteria for entanglement-annihilating
	maps~\cite{filippov-ziman-2013,filippov-rybar-ziman-2012} to find
	corresponding criteria for positive maps. Particular results of
	that kind are found in Sec.~\ref{section-unital-general}.
	
	\begin{proposition}
		\label{proposition-eb-otimes-positive} If $\Phi_1$ is entanglement
		breaking and $\Phi_2$ is positive, then the map
		$\Phi_1\otimes\Phi_2$ is positive.
	\end{proposition}
	\begin{proof}
		Since $\Phi_1$ is entanglement breaking, the operator
		$(\Phi_1\otimes{\rm Id})[R] = \sum_{k} R_k^{(1)} \otimes
		R_k^{(2)}$ is separable for any positive semidefinite $R$. Then,
		$(\Phi_1\otimes\Phi_2)[R] = \sum_{k} R_k^{(1)} \otimes
		\Phi_2[R_k^{(2)}] \geqslant 0$ in view of positivity of $\Phi_2$.
	\end{proof}
	
	
	\section{\label{section-depolarizing} Depolarizing qubit maps}
	To illustrate the problem of positivity of tensor product maps,
	let us consider an exactly solvable case of depolarizing qubit
	maps. The action of a depolarizing qubit map ${\cal D}_{q}$ is
	defined as follows:
	\begin{equation}
	{\cal D}_{q}[X] = q X + (1-q) {\rm tr}[X] \frac{1}{2} I,
	\end{equation}
	
	\noindent The map ${\cal D}_{q}$ is known to be positive if $q \in
	[-1,1]$ and completely positive if $q \in [-\frac{1}{3},1]$ (see,
	e.g.,~\cite{king-ruskai-2001,ruskai-2002}). In what follows, we
	analyze when the two-qubit map ${\cal D}_{q_1}\otimes{\cal
		D}_{q_2}$ is positive. Entanglement-annihilating properties of the
	map ${\cal D}_{q_1}\otimes{\cal D}_{q_2}$ and their
	generalizations (acting in higher dimensions) are considered in
	papers~\cite{filippov-ziman-2013,filippov-2014,lami-huber-2016}.
	
	Due to the convex structure of positive operators, if $({\cal
		D}_{q_1}\otimes{\cal D}_{q_2}) [\ket{\psi}\bra{\psi}] \geqslant 0$
	for all $\ket{\psi} \in {\cal H}_2 \otimes {\cal H}_2$, then the
	map ${\cal D}_1\otimes{\cal D}_2$ is positive. Since the norm of a
	vector $\ket{\psi}$ is not relevant for the analysis of
	positivity, let us consider pure input states
	$\omega=\ket{\psi}\bra{\psi}$ with $\ip{\psi}{\psi} = 1$. We use
	the Schmidt decomposition
	$\ket{\psi}=\sqrt{p}\ket{\phi\otimes\chi}+\sqrt{p_\perp}\ket{\phi_{\perp}\otimes\chi_{\perp}}$,
	where $\{\ket{\phi},\ket{\phi_{\perp}}\}$ and
	$\{\ket{\chi},\ket{\chi_{\perp}}\}$ are suitable orthonormal bases
	in Hilbert spaces of the first and second qubits, respectively,
	and $p$ and $p_{\perp}$ are real non-negative numbers such that
	$p+p_{\perp}=1$.
	
	\begin{figure}
		\includegraphics[width=8cm]{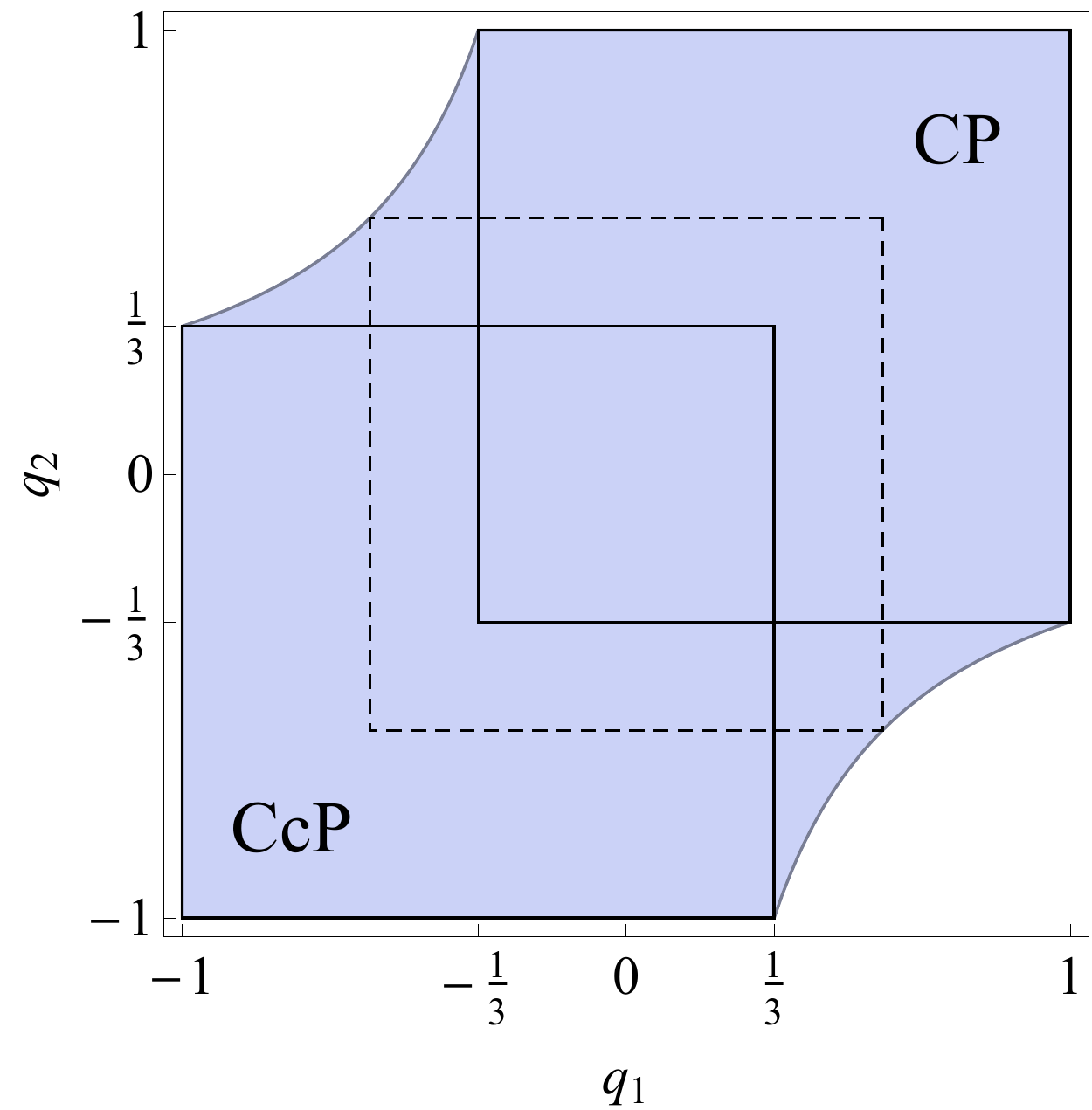}
		\caption{\label{figure-depolarizing} Shaded area is the region of
			parameters $q_1$ and $q_2$, where the map ${\cal
				D}_{q_1}\otimes{\cal D}_{q_2}$ is positive. Solid line regions
			correspond to completely positive (CP) and completely co-positive
			(CcP) maps ${\cal D}_{q_1}\otimes{\cal D}_{q_2}$.
			Proposition~\ref{proposition-unital-positive} detects positivity
			of the map ${\cal D}_{q_1}\otimes{\cal D}_{q_2}$ inside the dashed
			line region.}
	\end{figure}
	
	Action of the two-qubit map ${\cal D}_1\otimes{\cal D}_2$ on
	$\omega$ yields
	\begin{eqnarray}
	&&\omega_{\rm out} = ({\cal D}_{q_1}\otimes{\cal D}_{q_2})
	[\omega] = q_1 q_2 \omega +
	\frac{1}{2} (1-q_1) q_2 I \otimes \omega_2  \nonumber  \\
	&& + \frac{1}{2} q_1 (1-q_2) \omega_1\otimes I + \frac{1}{4}
	(1-q_1)(1-q_2) I \otimes I,
	\end{eqnarray}
	with the reduced states $\omega_1 = p \ket{\phi}\bra{\phi} +
	p_{\perp} \ket{\phi_{\perp}}\bra{\phi_{\perp}}$ and $\omega_2 = p
	\ket{\chi}\bra{\chi} + p_{\perp}
	\ket{\chi_{\perp}}\bra{\chi_{\perp}}$. The condition $\omega_{\rm
		out} \geqslant 0 $ reduces to
	\begin{equation}
	\label{depolarizing-intermediate-condition}
	\left(%
	\begin{array}{cccc}
	A_{+} + B_{+} & 0 & 0 & C \\
	0 & A_{-} + B_{-} & 0 & 0 \\
	0 & 0 & A_{-} - B_{-} & 0 \\
	C & 0 & 0 & A_{+} - B_{+} \\
	\end{array}%
	\right) \geqslant 0,
	\end{equation}
	
	\noindent where $A_{\pm} = 1 \pm q_1 q_2$, $B_{\pm} = (2p-1)(q_1
	\pm q_2)$, and $C = 4 \sqrt{p(1-p)} \, q_1 q_2$. After some
	algebra, we obtain that the condition
	\eqref{depolarizing-intermediate-condition} holds true for all $0
	\leqslant p \leqslant 1$ if
	\begin{equation}
	\label{depolarizing-condition} q_1 q_2 \geqslant -\frac{1}{3},
	\qquad -1 \leqslant q_1 \leqslant 1, \qquad -1 \leqslant q_2
	\leqslant 1.
	\end{equation}
	
	Inequalities~\eqref{depolarizing-condition} define the conditions
	under which the two-qubit map ${\cal D}_{q_1}\otimes{\cal
		D}_{q_2}$ is positive. We depict the corresponding area of
	parameters $q_1$ and $q_2$ in Fig.~\ref{figure-depolarizing}. Note
	that ${\cal D}_{q_1}\otimes{\cal D}_{q_2}$ is completely positive
	if $-\frac{1}{3} \leqslant q_1, q_2 \leqslant 1$. Analogously,
	${\cal D}_{q_1}\otimes{\cal D}_{q_2}$ is completely co-positive if
	$-1 \leqslant q_1, q_2 \leqslant \frac{1}{3}$.
	
	Let us demonstrate the use of
	Proposition~\ref{proposition-positive-one-sided}. Consider the
	reduction map ${\cal R}[X] = {\rm tr}[X] I - X$, which is known to
	be positive in qubit case~\cite{horodecki-1999}. The concatenation
	${\cal R} \circ {\cal D}_{q} = {\cal D}_{-q}$, i.e. the
	depolarizing map with parameter $-q$. The map ${\cal
		D}_{q_1}\otimes{\cal D}_{-q_2}$ is known to be positive
	entanglement annihilating if $q_1(-q_2) \leqslant \frac{1}{3}$ and
	$-1 \leqslant q_1, q_2 \leqslant 1$. According to
	Proposition~\ref{proposition-positive-one-sided}, these relations
	are sufficient for positivity of the map ${\cal D}_{q_1} \otimes
	({\cal R}\circ{\cal D}_{-q_2}) = {\cal D}_{q_1}\otimes{\cal
		D}_{q_2}$. In this particular case, these relations turn out to be
	the same as the necessary and sufficient conditions
	\eqref{depolarizing-condition}.
	
	
	\section{\label{section-unital-general} Unital qubit maps}
	
	A unital qubit map $\Phi$ satisfies $\Phi[I]=I$ and can be
	expressed in the form~\cite{king-ruskai-2001,ruskai-2002}
	\begin{equation}
	\label{Phi-through-Upsilon} \Phi[X] = W(\Upsilon[V X
	V^{\dag}])W^{\dag},
	\end{equation}
	
	\noindent where $V$ and $W$ are appropriate unitary operators such
	that the map $\Upsilon$ has the Pauli form, i.e.
	\begin{equation}
	\label{Upsilon}
	\Upsilon[X]=\frac{1}{2}\sum_{j=0}^{3}\lambda_{j}{\rm
		tr}[\sigma_{j}X]\sigma_{j} = \sum_{j=0}^3 q_j \sigma_j X \sigma_j,
	\end{equation}
	
	\noindent where $\sigma_0=I$ and $\{\sigma_{i}\}_{i=1}^{3}$ is a
	conventional set of Pauli operators. Thus, up to a unitary
	preprocessing $(V \cdot V^{\dag})$ and postprocessing $(W \cdot
	W^{\dag})$ the unital map $\Phi$ reduces to the map $\Upsilon$.
	From Eq.~\eqref{Phi-through-Upsilon} it is not hard to see that
	the two-qubit unital map $\Phi_1 \otimes \Phi_2$ is positive if
	and only if $\Upsilon_1 \otimes \Upsilon_2$ is positive. In this
	section, we will consider properties of maps $\Upsilon_1 \otimes
	\Upsilon_2$.
	
	The relation between parameters $\{\lambda_j\}$ and $\{q_j\}$ in
	formula~\eqref{Upsilon} is given by
	\begin{equation}
	\left(%
	\begin{array}{c}
	q_0 \\
	q_1 \\
	q_2 \\
	q_3 \\
	\end{array}%
	\right) = \frac{1}{4} \left(%
	\begin{array}{cccc}
	1 & 1  & 1  & 1 \\
	1 & 1  & -1 & -1 \\
	1 & -1 & 1  & -1 \\
	1 & -1 & -1 & 1 \\
	\end{array}%
	\right) \left(%
	\begin{array}{c}
	\lambda_0 \\
	\lambda_1 \\
	\lambda_2 \\
	\lambda_3 \\
	\end{array}%
	\right),
	\end{equation}
	
	\noindent i.e. ${\bf q} = \frac{1}{2}H \boldsymbol{\lambda}$,
	where ${\bf q} = (q_0,q_1,q_2,q_3)^{\top}$, $H$ is the $4 \times
	4$ Hadamard matrix, and $\boldsymbol{\lambda} =
	(\lambda_0,\lambda_1,\lambda_2,\lambda_3)^{\top}$.
	
	For hermicity-preserving maps $\Upsilon$ the parameters
	$\{\lambda_j\}$ and $\{q_j\}$ are real. Positive maps correspond
	to parameters $\lambda_0 \geqslant 0$, $-\lambda_0 \leqslant
	\lambda_1,\lambda_2,\lambda_3 \leqslant \lambda_0$. Completely
	positive maps correspond to $q_j \geqslant 0$, $j=0,\ldots,3$.
	Trace preserving maps are those with $\lambda_0 = 1$. Completely
	positive trace preserving unital qubit maps are called unital
	qubit channels and are essentially the random unitary
	channels~\cite{audenaert-scheel-2008}.
	
	\begin{proposition}
		\label{proposition-unital-positive} A unital two-qubit map
		$\Upsilon_1 \otimes \Upsilon_2$ is positive if $\Upsilon_1^2$ and
		$\Upsilon_2^2$ are both entanglement breaking.
	\end{proposition}
	\begin{proof}
		Consider a map $\overline{\Upsilon}_2 = {\cal R} \circ
		\Upsilon_2$, where ${\cal R}$ is the qubit reduction map. Then
		$\overline{\Upsilon}_2^2 = \Upsilon_2^2$ is the entanglement
		breaking map. If $\Upsilon_1^2$ and $\overline{\Upsilon}_2^2$ are
		both entanglement breaking, then the map $\Upsilon_1 \otimes
		\overline{\Upsilon}_2$ is positive entanglement annihilating
		according to the Proposition~1 of
		Ref.~\cite{filippov-rybar-ziman-2012}. By
		Proposition~\ref{proposition-positive-one-sided} of the present
		paper, the map $\Upsilon_1 \otimes ({\cal R} \circ
		\overline{\Upsilon}_2) = \Upsilon_1 \otimes \Upsilon_2$ is
		positive.
	\end{proof}
	
	The ``power'' of Proposition~\ref{proposition-unital-positive} can
	be illustrated by the example of the local two-qubit depolarizing
	map ${\cal D}_{q_1} \otimes {\cal D}_{q_2}$. Since ${\cal D}_{q}^2
	= {\cal D}_{q^2}$, the maps ${\cal D}_{q_1}^2$ and ${\cal
		D}_{q_2}^2$ are both entanglement breaking if $q_1^2, q_2^2
	\leqslant \frac{1}{3}$~\cite{ruskai-2003}, i.e.
	$-\frac{1}{\sqrt{3}} \leqslant q_1, q_2 \leqslant
	\frac{1}{\sqrt{3}}$. Corresponding region of parameters is
	depicted in Fig.~\ref{figure-depolarizing}. Clearly,
	Proposition~\ref{proposition-unital-positive} provides only a
	sufficient but not a necessary condition for positivity.
	
	Being applied to the unital map $\Upsilon\otimes\Upsilon$,
	Proposition~\ref{proposition-unital-positive} guarantees that the
	map $\Upsilon\otimes\Upsilon$ is positive if
	$\lambda_1^2+\lambda_2^2+\lambda_3^2 \leqslant \lambda_0^2$
	($\Upsilon^2$ is entanglement breaking). The corresponding ball
	(for $\lambda_0 = 1$) is depicted in
	Fig.~\ref{figure-geometry-simple}. Usual powers of linear maps
	(self-concatenations $\Upsilon^{n}$) also find applications in
	quantum information theory, for instance, in noise quantification
	\cite{de-pasquale-2012,lami-giovannetti-2015}.
	
	\begin{figure}
		\includegraphics[width=8cm]{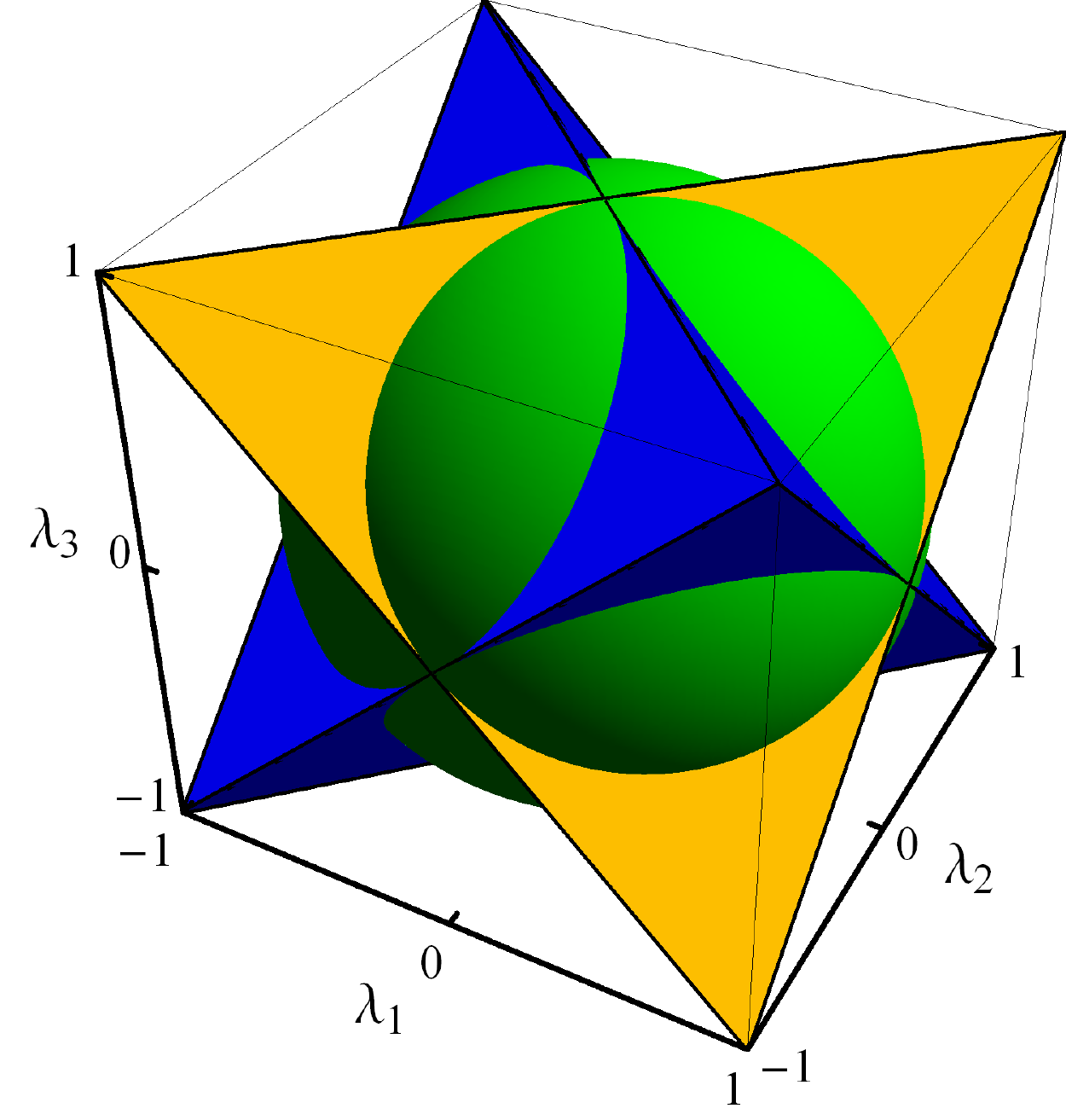}
		\caption{\label{figure-geometry-simple} Regions of parameters
			$\lambda_1,\lambda_2,\lambda_3$ determining particular properties
			of the unital qubit map $\Upsilon$ defined by Eq.~\eqref{Upsilon}.
			Tetrahedrons represent trivial tensor-stable positive maps (yellow
			one corresponds to completely positive maps, blue one corresponds
			to completely co-positive maps). Sphere corresponds to positive
			entanglement annihilating maps $\Upsilon\otimes\Upsilon$ (see
			Proposition~\ref{proposition-unital-positive}).}
	\end{figure}
	
	\begin{remark}
		Comparison of formulas \eqref{Phi-through-Upsilon} and
		\eqref{Upsilon} clarifies that, in general, $\Phi^2 \neq
		\Upsilon^2$. This fact is analogous to filtering $\Upsilon \circ
		{\cal U} \circ
		\Upsilon$~\cite{gavenda-2008,de-pasquale-2012,lami-giovannetti-2015},
		where the intermediate (unitary) map ${\cal U}$ is used to prevent
		$\Upsilon^2$ from becoming entanglement breaking.
	\end{remark}
	
	\begin{remark}
		In contrast to the entanglement annihilating property, which
		states that $\Upsilon_1 \otimes \Upsilon_2$ is positive
		entanglement annihilating if both $\Upsilon_1^2$ and
		$\Upsilon_2^2$ are entanglement breaking, complete positivity of
		the maps $\Upsilon_1^2$ and $\Upsilon_2^2$ does not imply that
		$\Upsilon_1 \otimes \Upsilon_2$ is positive. Counterexample is the
		case $\Upsilon_1 = {\rm Id}$ and $\Upsilon_2 = \top$.
	\end{remark}
	
	
	\section{\label{section-2-tensor-stable-positive-unital} 2-tensor-stable positive unital qubit maps}
	
	In this section, we analyze positivity of two-qubit unital maps
	$\Phi^{\otimes 2}$. By Eq.~\eqref{Phi-through-Upsilon}, $\Phi$ is
	2-tensor-stable positive if and only if $\Upsilon$ is
	2-tensor-stable positive. Without loss of generality one can
	impose the trace-preserving condition, $\lambda_0=1$, then the
	remaining three parameters $\{\lambda_j\}_{j=1}^{3}$ in formula
	\eqref{Upsilon} are scaling coefficients of the Bloch ball axes.
	Thus, the map $\Upsilon$ is given by a point in the Cartesian
	coordinate system $(\lambda_1,\lambda_2,\lambda_3)$ and can be
	readily visualized.
	
	\begin{figure}
		\includegraphics[width=8cm]{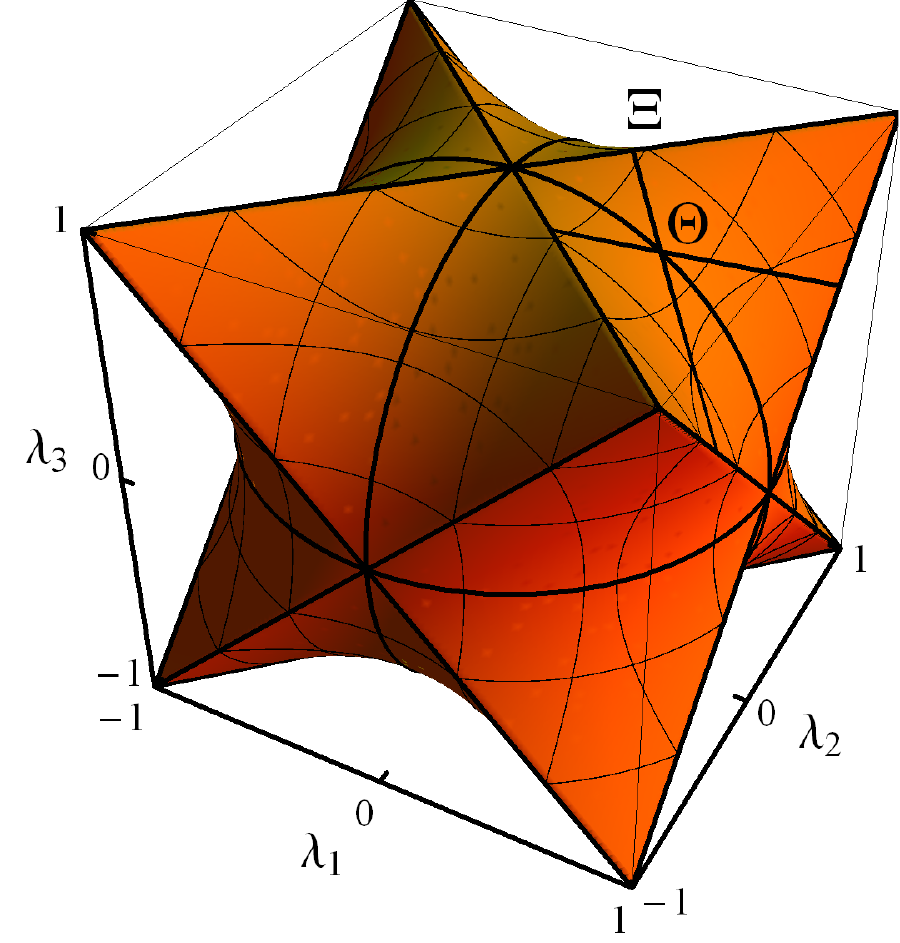}
		\caption{\label{figure2} Region of parameters
			$\lambda_1,\lambda_2,\lambda_3$, where the unital qubit map
			$\Upsilon$ (see Eq.~\eqref{Upsilon}) is 2-tensor-stable positive.
			Map $\Theta$ belongs to the gorge of a hyperboloid, line
			$\Theta\Xi$ is a generatrix of the hyperboloid surface.}
	\end{figure}
	
	\begin{proposition}
		\label{proposition-2-tsp} $\Upsilon$ is 2-tensor-stable positive
		if and only if $\Upsilon^2$ is completely positive.
	\end{proposition}
	\begin{proof}
		\textit{Necessity}. Let us prove that if $\Upsilon\otimes\Upsilon$
		is positive, then $\Upsilon^2$ is completely positive. Suppose
		$(\Upsilon\otimes\Upsilon)[R] \geqslant 0$ for all $R \geqslant
		0$. Let $R$ be equal to $\ket{\psi_+}\bra{\psi_+}$, where
		$\ket{\psi_+}=\frac{1}{\sqrt{2}}(\ket{00}+\ket{11})$. It is not
		hard to see that
		\begin{eqnarray}
		\label{choi-of-upsilon-squared}
		&& (\Upsilon\otimes\Upsilon) [\ket{\psi_+} \bra{\psi_+}] \nonumber\\
		&& = \frac{1}{4}
		\left(%
		\begin{array}{cccc}
		1+\lambda_3^2 & 0 & 0 & \lambda_1^2 + \lambda_2^2 \\
		0 & 1-\lambda_3^2 & \lambda_1^2-\lambda_2^2 & 0 \\
		0 & \lambda_1^2-\lambda_2^2 & 1-\lambda_3^2 & 0 \\
		\lambda_1^2 + \lambda_2^2 & 0 & 0 & 1+\lambda_3^2 \\
		\end{array}%
		\right) = \Omega_{\Upsilon^2}, \qquad
		\end{eqnarray}
		
		\noindent i.e. $\Omega_{\Upsilon^2} \geqslant 0$ and the map
		$\Upsilon^2$ is completely positive.
		
		Taking into account the explicit form of the Choi matrix
		\eqref{choi-of-upsilon-squared}, we get
		\begin{equation}
		\label{hyperboloids} \{\Upsilon^{2} \ \text{is CP}\}
		\Leftrightarrow \left\{
		\begin{array}{l}
		1 + \lambda_1^2 \geqslant \lambda_2^2 + \lambda_3^2, \\
		1 + \lambda_2^2 \geqslant \lambda_1^2 + \lambda_3^2, \\
		1 + \lambda_3^2 \geqslant \lambda_1^2 + \lambda_2^2, \\
		\end{array}
		\right.
		\end{equation}
		
		\noindent or, concisely, $ 1 \pm \lambda_{3}^2 \geqslant
		|\lambda_{1}^2\pm \lambda_{2}^2|$. Each of inequalities
		\eqref{hyperboloids} defines an interior of the one-sheet
		hyperboloid in the space of parameters
		$(\lambda_1,\lambda_2,\lambda_3)$. Intersection of these three
		hyperboloids is depicted in Fig.~\ref{figure2}. Gorges (throats)
		of those hyperboloids are exactly three mutually perpendicular
		great circles (orthodromes) of a unit sphere.
		
		\textit{Sufficiency}. Suppose $\Upsilon^2$ is completely positive
		and $\lambda_0 = 1$, then parameters
		$(\lambda_1,\lambda_2,\lambda_3)$ satisfy the inequalities
		\eqref{hyperboloids}.
		
		Let us use the alternative description of the map $\Upsilon$,
		namely, $\Upsilon [X] = \sum_{i=0}^3 q_i \sigma_i X \sigma_i$.
		Then $(\Upsilon \otimes \Upsilon) [X] = \sum_{i,j=0}^3 q_i q_j
		\sigma_i \otimes \sigma_j X \sigma_i \otimes \sigma_j$. To
		demonstrate positivity of the map $\Upsilon \otimes \Upsilon$, it
		suffices to show that $\bra{\varphi} (\Upsilon \otimes \Upsilon)
		[\ket{\psi}\bra{\psi}] \ket{\varphi} \geqslant 0$ for all
		$\ket{\psi},\ket{\varphi} \in {\cal H}_2 \otimes {\cal H}_2$. We
		have
		\begin{eqnarray}
		\bra{\varphi} (\Upsilon \otimes \Upsilon) [\ket{\psi}\bra{\psi}]
		\ket{\varphi} &=& \sum_{i,j=0}^3 q_i q_j \left| \bra{\varphi}
		\sigma_i \otimes \sigma_j \ket{\psi} \right|^2 \nonumber\\
		&=& {\bf q}^{\top} A {\bf q} = \frac{1}{4}
		\boldsymbol{\lambda}^{\top} H^{\dag} A H \boldsymbol{\lambda},
		\end{eqnarray}
		
		\noindent where the matrix elements
		\begin{equation}
		A_{ij} = \frac{1}{2}\left( \left| \bra{\varphi} \sigma_i \otimes
		\sigma_j \ket{\psi} \right|^2 + \left| \bra{\varphi} \sigma_j
		\otimes \sigma_i \ket{\psi} \right|^2 \right) \geqslant 0.
		\end{equation}
		
		Thus, the symmetric matrix $A$ has non-negative entries only and,
		according to the Perron-Frobenius theorem, the absolute value of
		its minimal eigenvalue, $|\lambda_{-}|$, cannot exceed its maximal
		eigenvalue, $\lambda_{+} > 0$ (see, e.g.,~\cite{gantmacher-2000}).
		Since the Hadamard matrix $H$ is unitary, the eigenvalues of
		matrices $H^{\dag} A H$ and $A$ coincide. It means that the
		absolute value $|\lambda_{-}|$ of any negative coefficient
		$\lambda_{-}$ in the diagonal representation of the quadratic form
		$\boldsymbol{\lambda}^{\top} H^{\dag} A H \boldsymbol{\lambda}$ is
		less or equal than the maximal positive coefficient. Consequently,
		the principal curvatures $k_1$ and $k_2$ of a quadric surface
		$\boldsymbol{\lambda}^{\top} H^{\dag} A H \boldsymbol{\lambda} =
		0$ satisfy $-1 \leqslant k_1,k_2 \leqslant 1$. On the other hand,
		each of equalities \eqref{hyperboloids} defines a surface with
		boundary principal curvatures ($\min k_1 = -1$, $\max k_2 = 1$).
		Thus, no quadric $\boldsymbol{\lambda}^{\top} H^{\dag} A H
		\boldsymbol{\lambda}=0$ can intersect the interior region of all
		inequalities \eqref{hyperboloids} without intersecting the regions
		of completely positive maps or completely co-positive maps (two
		tetrahedrons in Fig.~\ref{figure-geometry-simple}). Roughly
		speaking, all quadric surfaces $\boldsymbol{\lambda}^{\top}
		H^{\dag} A H \boldsymbol{\lambda}=0$ are more ``flat'' than those
		of Eqs.~\eqref{hyperboloids}. Therefore, the interior region of
		all inequalities \eqref{hyperboloids} is the interior set of all
		figures $\boldsymbol{\lambda}^{\top} H^{\dag} A H
		\boldsymbol{\lambda} \geqslant 0$, which implies $\{\Upsilon^2$ is
		CP$\} \Rightarrow \{ \Upsilon \otimes \Upsilon$ is positive$\}$.
	\end{proof}
	
	Inequalities \eqref{hyperboloids} specify a non-convex geometrical
	figure in the space of parameters
	$(\lambda_1,\lambda_2,\lambda_3)$. However, any interior point of
	that figure corresponds to a convex sum of some boundary map
	$\Upsilon$ and the completely depolarizing map $\mathcal{D}_0$ and
	corresponds to a 2-tensor-stable positive map since the boundary
	map does so. Parameterizing the surface of hyperboloids, one can
	also find numerical evidence of
	Proposition~\ref{proposition-2-tsp}.
	
	Any one-sheet hyperboloid is doubly ruled, i.e. it has two
	distinct generatrices that pass trough every point. Without loss
	of generality, let us consider a particular hyperboloid fragment
	with vertices $(1,1,1)$, $(0,0,1)$, $(1,-1,1)$, and $(1,0,0)$ in
	the space $(\lambda_1,\lambda_2,\lambda_3)$. These vertices
	correspond to the maps ${\rm Id}$, $\mathcal{Z}$, ${\top}$, and
	$\mathcal{X}$, respectively. The first family of generatrices is
	given by straight lines passing through points $\left( x,x,1
	\right)$ and $\left( 1,-f(x),f(x) \right)$, where $f(x) =
	(1-x)/(1+x)$ and $x\in (0,1)$. The second family of generatrices
	is formed by straight lines passing through points $\left( y,-y,1
	\right)$ and $\left( 1,f(y),f(y) \right)$, $y\in (0,1)$. Any point
	inside the involved hyperboloid fragment is defined by a pair of
	parameters $(x,y)\in[0,1]^2$ and corresponds to the following map:
	\begin{equation}
	\label{Upsilon-x-y} \Upsilon = \frac{x(1-y){\rm Id} + y(1-x){\top}
		+ 2xy \mathcal{X} + (1-x)(1-y) \mathcal{Z}}{1+xy},
	\end{equation}
	
	\noindent whose parameters read
	\begin{equation}
	\label{lambda-x-y} \lambda_{1} = \frac{x+y}{1 + x y}, \quad
	\lambda_{2} = \frac{x - y}{1 + x y}, \quad \lambda_{3} = \frac{1 -
		x y}{1 + x y}.
	\end{equation}
	
	Then one can check block-positivity of the Choi matrix of the map
	$\Upsilon\otimes\Upsilon$, i.e. to validate inequality
	\begin{equation}
	\label{choi-block-positive} \bra{\varphi^{AB}} \otimes
	\bra{\chi^{A'B'}} (\Omega_{\Upsilon}^{AA'} \otimes
	\Omega_{\Upsilon}^{BB'}) \ket{\varphi^{AB}} \otimes
	\ket{\chi^{A'B'}} \geqslant 0
	\end{equation}
	
	\noindent numerically for all two qubit states
	$\ket{\varphi^{AB}}$ and $\ket{\chi^{AB}}$ and all $0 \leqslant
	x,y \leqslant 1$.
	
	Alternatively, positivity of the operator $\varrho =
	(\Upsilon\otimes\Upsilon) [\ket{\psi}\bra{\psi}]$ is guaranteed by
	the requirement that all coefficients $s_i$, $i=1,\ldots,4$ of the
	characteristic polynomial
	\begin{equation}
	\label{characteristic-polynomial} {\rm det}(\lambda I-\varrho) =
	\lambda^N - s_1 \lambda^{N-1} + s_2 \lambda^{N-2} - \ldots +
	(-1)^N s_N
	\end{equation}
	
	\noindent are non-negative, i.e. $\rho\geqslant 0 \Leftrightarrow
	s_k \geqslant 0$, $k=1,\ldots,4$~\cite{bengtsson-zyczkowski}. The
	coefficients $s_1={\rm tr}[\varrho]$, $s_2=\frac{1}{2}(s_1{\rm
		tr}[\varrho]-{\rm tr}[\varrho^2])$, and in general, iteratively,
	$s_k=\frac{1}{k}(s_{k-1}{\rm tr}[\rho]-s_{k-2}{\rm tr}[\rho^2]+
	\ldots +(-1)^{k-1}{\rm tr}[\rho^{k}])$. Then one can numerically
	check positivity of matrices
	$(\Upsilon\otimes\Upsilon)[\ket{\psi}\bra{\psi}]$ for all
	two-qubit states $\ket{\psi}$ and parameters $0 \leqslant x,y
	\leqslant 1$.
	
	
	\section{\label{section-decomposable} Decomposability}
	
	Following the results of Refs.~\cite{stormer-1963,stormer-1982},
	we will refer to the map of the form $\Phi_1 + \Phi_2$, where
	$\Phi_1$ is completely positive and $\Phi_2$ is completely
	co-positive, as decomposable.
	
	All positive qubit maps $\Phi: {\cal B}({\cal H}_2) \mapsto {\cal
		B}({\cal H}_2)$ are known to be decomposable~\cite{stormer-1963}.
	Decomposability of extremal positive unital maps on $M_2$ is
	analyzed in Ref.~\cite{majewski-2005}. However, even if $\Phi$ is
	decomposable, it does not imply that $\Phi^{\otimes 2}$ is
	decomposable as it contains terms $\Phi_1 \otimes \Phi_2$ and
	$\Phi_2 \otimes \Phi_1$ which are not necessarily positive.
	Moreover, there exist examples of indecomposable maps ${\cal
		B}({\cal H}_2) \mapsto {\cal B}({\cal
		H}_4)$~\cite{woronowicz-1976}. This means that the decomposability
	of positive tensor powers $\Phi \otimes \Phi$ (as well as the more
	general property of $k$-decomposability~\cite{labuschagne-2006})
	is still an open problem even for qubit maps $\Phi$. In what
	follows, we make some steps toward understanding of this problem
	and consider examples of non-trivial 2-tensor-stable unital maps
	$\Upsilon$ such that $\Upsilon\otimes\Upsilon$ is decomposable.
	
	\begin{example}
		Positive map $\Upsilon \otimes \Upsilon$ with parameters
		$\lambda_1=\frac{1}{\sqrt{2}}$, $\lambda_2=0$,
		$\lambda_3=\frac{1}{\sqrt{2}}$ is decomposable.
		
		In fact, it is not hard to see that
		\begin{equation}
		\Upsilon \otimes \Upsilon = \frac{1}{2} F \circ \left( {\rm Id}
		\otimes {\rm Id} + \top \otimes \top \right),
		\end{equation}
		
		\noindent where $F$ is a two-qubit map of the form $F[X] =
		\frac{1}{4} \sum_{i,j=0}^3 \lambda_{ij} {\rm tr}[\sigma_i \otimes
		\sigma_j X] \sigma_i \otimes \sigma_j$ with $\lambda_{00} = 1$,
		$\lambda_{01} = \lambda_{03} = \lambda_{10} = \lambda_{30} =
		\frac{1}{\sqrt{2}}$, $\lambda_{02} = \lambda_{11} = \lambda_{13} =
		\lambda_{20} = \lambda_{31} = \lambda_{33} = \frac{1}{2}$,
		$\lambda_{12} = \lambda_{21} = \lambda_{23} = \lambda_{32} =
		\frac{1}{4}$, and $\lambda_{22} = 0$. Eigenvalues of the Choi
		matrix $\Omega_F$ are all non-negative, consequently, $F$ is
		completely positive, $F \circ (\top \otimes \top)$ is completely
		co-positive, and $\Upsilon \otimes \Upsilon$ is decomposable.
	\end{example}
	
	\begin{example}
		Consider the one-parametric family of maps $(\mu \Upsilon_1 +
		(1-\mu) \Upsilon_2)^{\otimes 2}$, where $\Upsilon_1$ is given by
		parameters $\lambda_0 = 1$, $\lambda_1 = \lambda_2 = \lambda_3 =
		\frac{2}{3}$, $\Upsilon_2$ is given by parameters $\lambda_0 =
		\lambda_3 = 1$, $\lambda_1 = - \lambda_2 = \frac{1}{20}$. Let us
		show that $(\mu \Upsilon_1 + (1-\mu) \Upsilon_2)^{\otimes 2}$ is
		positive and decomposable for all $0 \leqslant \mu \leqslant 1$.
		Note that $\mu \Upsilon_1 + (1-\mu) \Upsilon_2$ is non-trivial
		2-tensor-stable positive for $0<\mu<\frac{3}{13}$.
		
		The map $\Upsilon_1^{\otimes 2}$ is completely positive, and the
		map $\Upsilon_2^{\otimes 2}$ is completely co-positive. Let us
		note that
		\begin{equation}
		\Upsilon_1 \otimes \Upsilon_2 = F \circ \left( \frac{3}{4} G_1
		\otimes {\rm Id} + \frac{1}{4} \top \otimes G_2  \right),
		\end{equation}
		
		\noindent where $G_1 = {\cal Z}$ ($\lambda_0 = \lambda_3 = 1$,
		$\lambda_1 = \lambda_2 = 0$), $G_2 = {\cal D}_{1/3}$ (depolarizing
		map with parameter $1/3$), and $F$ is a two-qubit map $F[X] =
		\frac{1}{4} \sum_{i,j=0}^3 \lambda_{ij} {\rm tr}[\sigma_i \otimes
		\sigma_j X] \sigma_i \otimes \sigma_j$ with
		$\lambda_{0n}=\lambda_{3n}=(4+\delta_{n0})/5$,
		$\lambda_{1n}=\lambda_{2n}=(2-\delta_{n0})/5$. Since the eigenvalues of
		the Choi matrix $\Omega_F$ are all non-negative, $F \circ (G_1
		\otimes {\rm Id})$ is completely positive as the concatenation of
		completely positive maps, and $F \circ (\top \otimes G_2)$ is
		completely co-positive as the concatenation of completely positive
		and completely co-positive maps. (Note that $G_2$ is entanglement
		breaking and $\top \otimes G_2$ is completely co-positive.) Thus,
		$\Upsilon_1 \otimes \Upsilon_2$ is decomposable. Analogously,
		$\Upsilon_2 \otimes \Upsilon_1$ is decomposable. Finally, $(\mu
		\Upsilon_1 + (1-\mu) \Upsilon_2)^{\otimes 2}$ is decomposable as
		the convex sum of decomposable maps.
	\end{example}
	
	The examples above stimulate us to make a conjecture that all
	positive unital two-qubit maps of the form $\Upsilon \otimes
	\Upsilon$ are decomposable.
	
	\begin{figure}
		\includegraphics[width=8cm]{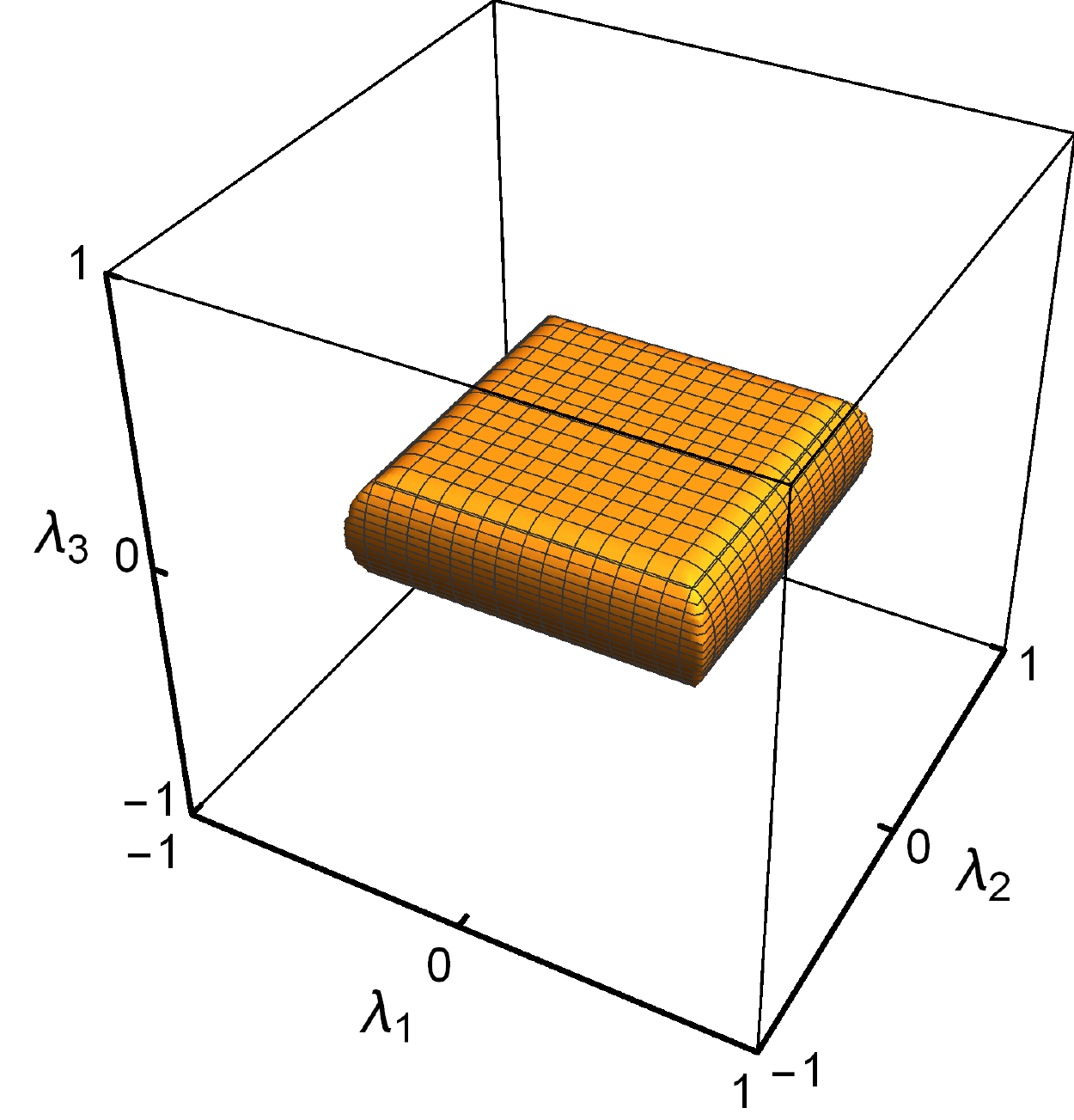}
		\caption{\label{figure-non-unital-positive} Non-unital map
			\eqref{matrix-form-non-unital} is positive for $t=0.8$ in the
			shaded region of parameters $\lambda_1$, $\lambda_2$,
			$\lambda_3$.}
	\end{figure}

	
	\section{\label{section-2-tensor-stable-positive-non-unital} Non-unital qubit maps}
	Similarly to Eq.~\eqref{Phi-through-Upsilon}, an interior map of
	the cone of positive non-unital qubit maps $\Phi: ({\cal B}({\cal
		H}_2))^{+} \mapsto ({\cal B}({\cal H}_2))^{+}$ can be represented
	in the form of the following
	concatenation~\cite{gurvits-2004,aubrun-szarek-2015}:
	\begin{equation}
	\label{Phi-through-Upsilon-non-unital} \Phi[X] = B(\Upsilon[A X
	A^{\dag}])B^{\dag},
	\end{equation}
	
	\noindent where $\Upsilon$ is given by Eq.~\eqref{Upsilon} and
	$A,B \in {\cal B}({\cal H}_2)$ are positive-definite operators. As
	$A$ and $B$ are non-degenerate, the condition $\bra{\varphi}
	(\Phi_1\otimes\Phi_2)[\ket{\psi}\bra{\psi}] \ket{\varphi}
	\geqslant 0$ holds for all $\ket{\psi},\ket{\varphi}\in{\cal H}_4$
	if and only if $\bra{\widetilde{\varphi}}
	(\Upsilon_1\otimes\Upsilon_2)[\ket{\widetilde{\psi}}\bra{\widetilde{\psi}}]
	\ket{\widetilde{\varphi}} \geqslant 0$ holds for all
	$\ket{\widetilde{\psi}},\ket{\widetilde{\varphi}}\in{\cal H}_4$,
	since $\ket{\widetilde{\psi}} = A_1 \otimes A_2 \ket{\psi}$ and
	$\ket{\widetilde{\varphi}} = B_1^{\dag} \otimes B_2^{\dag}
	\ket{\varphi}$. Thus, the positivity of a tensor product of
	non-unital maps $\Phi_1 \otimes \Phi_2$ is equivalent to the
	positivity of the tensor product of corresponding unital maps
	$\Upsilon_1 \otimes \Upsilon_2$.
	
	A qubit map $\Phi$ can be expressed in an appropriate basis by its
	matrix form ${\cal E}_{ij} = \frac{1}{2}{\rm tr}\left[ \sigma_i
	\Phi[\sigma_j] \right]$ as
	follows~\cite{king-ruskai-2001,ruskai-2002}:
	\begin{equation}
	\label{matrix-form-non-unital-general}
	{\cal E} = \left(%
	\begin{array}{cccc}
	1 & 0 & 0 & 0 \\
	t_1 & \lambda_1 & 0 & 0 \\
	t_2 & 0 & \lambda_2 & 0 \\
	t_3 & 0 & 0 & \lambda_3 \\
	\end{array}%
	\right).
	\end{equation}
	
	\begin{figure}
		\includegraphics[width=8cm]{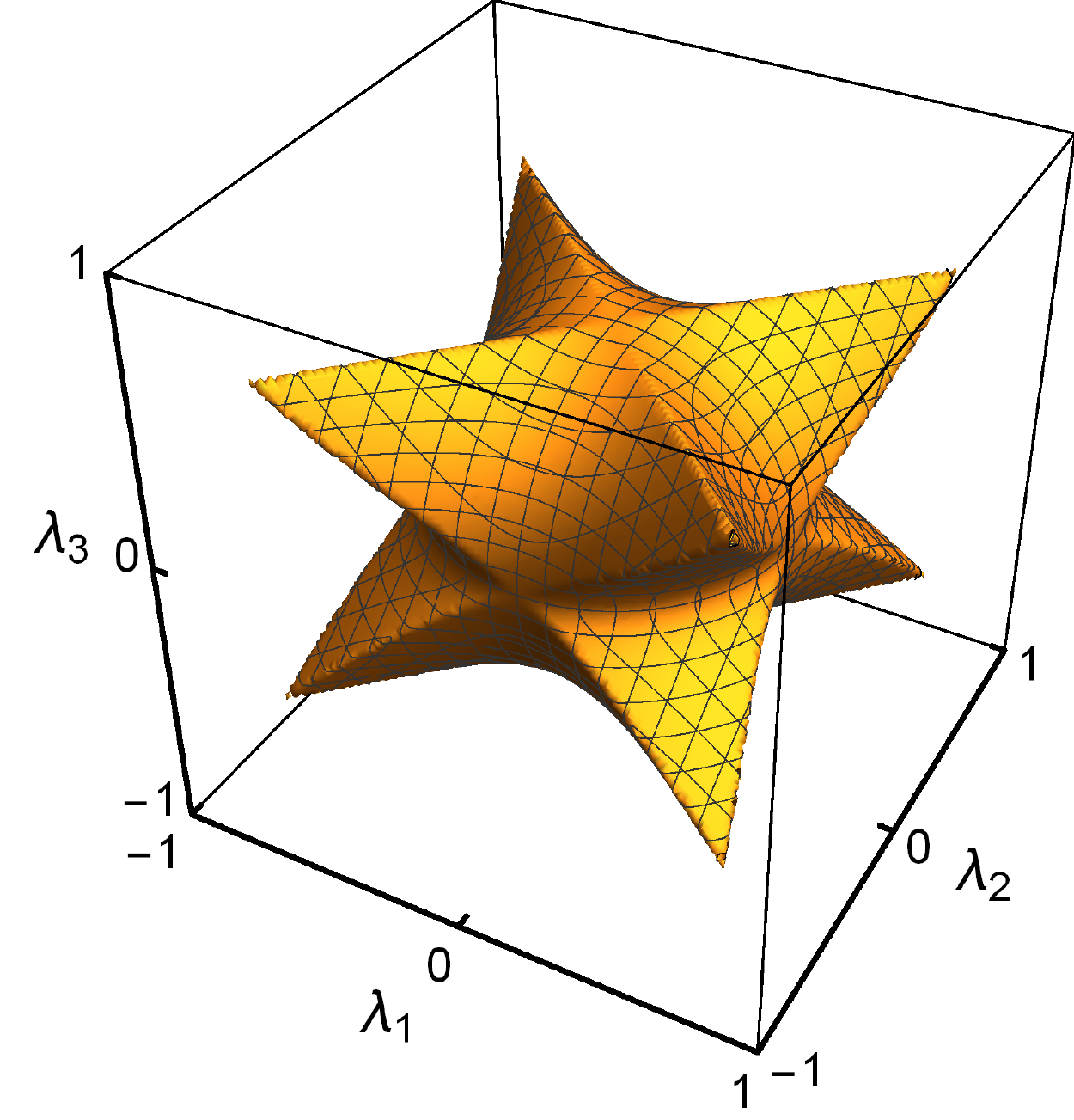}
		\caption{\label{figure-non-unital-entangled} Region of parameters
			$\lambda_1$, $\lambda_2$, $\lambda_3$, satisfying the condition
			$(\Phi\otimes\Phi)[\ket{\psi_+}\bra{\psi_+}] \geqslant 0$ for
			non-unital maps $\Phi$ defined by
			Eq.~\eqref{matrix-form-non-unital} with $t=0.8$.}
	\end{figure}
	
	To demonstrate the idea of reducing the problem to the case of
	unital maps, let us consider a four-parametric family of maps
	$\Phi$ whose matrix representation reads
	\begin{equation}
	\label{matrix-form-non-unital}
	{\cal E} = \left(%
	\begin{array}{cccc}
	1 & 0 & 0 & 0 \\
	0 & \lambda_1 & 0 & 0 \\
	0 & 0 & \lambda_2 & 0 \\
	t & 0 & 0 & \lambda_3 \\
	\end{array}%
	\right).
	\end{equation}
	
	\noindent Such a family comprises the description of extremal
	completely positive qubit maps~\cite{ruskai-2002}.
	
	\begin{proposition}
		\label{proposition-reduction-to-unitary} Let $\Phi$ be a map
		defined by the matrix
		representation~\eqref{matrix-form-non-unital} with
		$1-|t|-|\lambda_3| > 0$ and
		\begin{eqnarray}
		&& \!\!\!\!\! A^{-1} = \left(%
		\begin{array}{cc}
		a_{+} b_{-} & 0 \\
		0 & a_{-} b_{+} \\
		\end{array}%
		\right), \quad B^{-1} = \frac{1}{2} \left(%
		\begin{array}{cc}
		b_{-} & 0 \\
		0 & b_{+} \\
		\end{array}%
		\right), \quad \\
		&& \!\!\!\!\! a_{\pm} = \sqrt{1 \pm t - \lambda_3}, \quad b_{\pm}
		= \pm \sqrt[4]{(1 \pm t)^2 - \lambda_3^2},
		\end{eqnarray}
		
		\noindent then the map $\widetilde{\Upsilon}[Y] = B^{-1}
		\Phi[A^{-1} Y (A^{-1})^{\dag}] (B^{-1})^{\dag}$ is proportional to
		a unital map and the corresponding coefficients in
		Eq.~\eqref{Upsilon} equal
		\begin{eqnarray}
		&& \!\!\!\!\!\!\!\!\!\!\!\! \widetilde{\lambda}_0 \!=\!
		\frac{1}{2} \big( (1 -
		\lambda_3)^2 - t^2 \big) \nonumber\\
		&& \!\!\!\!\!\!\!\!\!\!\!\! \times \left[ (1 \!+\! \lambda_3)^2
		\!-\! t^2 + \sqrt{\big( (1 \!-\! t)^2 \!-\! \lambda_3^2 \big)
			\big( (1
			\!+\!t)^2 \!-\! \lambda_3^2 \big)} \right], \\
		&& \!\!\!\!\!\!\!\!\!\!\!\! \widetilde{\lambda}_1 \!=\! \lambda_1 \sqrt{ \big( (1 \! - \! \lambda_3)^2 \!-\! t^2 \big) \big((1 \!-\! t)^2 \!-\! \lambda_3^2 \big) \big( (1 \!+\! t)^2 \!-\! \lambda_3^2 \big)}, \\
		&& \!\!\!\!\!\!\!\!\!\!\!\! \widetilde{\lambda}_2 \!=\! \lambda_2 \sqrt{ \big( (1 \! - \! \lambda_3)^2 \!-\! t^2 \big) \big((1 \!-\! t)^2 \!-\! \lambda_3^2 \big) \big( (1 \!+\! t)^2 \!-\! \lambda_3^2 \big)}, \\
		&& \!\!\!\!\!\!\!\!\!\!\!\! \widetilde{\lambda}_3 \!=\!
		\frac{1}{2} \big( (1
		- \lambda_3)^2 - t^2 \big) \nonumber\\
		&& \!\!\!\!\!\!\!\!\!\!\!\! \times \left[ (1 \!+\! \lambda_3)^2
		\!-\! t^2 - \sqrt{\big( (1 \!-\! t)^2 \!-\! \lambda_3^2 \big)
			\big( (1 \!+\!t)^2 \!-\! \lambda_3^2 \big)} \right].
		\end{eqnarray}
	\end{proposition}
	
	\begin{proof}
		Since the concatenation of maps corresponds to the product of
		their matrix representations, a straightforward calculation yields
		the map $\widetilde{\Upsilon}$ which has the form \eqref{Upsilon}.
	\end{proof}
	
	Thus, the map $\Phi$ given by Eq.~\eqref{matrix-form-non-unital}
	is positive when (i) $1-|t|-|\lambda_3| > 0$ and
	$|\widetilde{\lambda}_k| \leqslant \widetilde{\lambda}_0$,
	$k=1,2,3$ or (ii) $1-|t|-|\lambda_3| = 0$ and
	$\lambda_1^2,\lambda_2^2 \leqslant 1-|t|$. The latter one can
	readily be obtained by considering Bloch ball transformations.
	Fixing parameter $t$, one can visualize conditions (i)--(ii) in
	the reference frame $(\lambda_1,\lambda_2,\lambda_3)$ (see Fig.
	\ref{figure-non-unital-positive}).
	
	\begin{figure}
		\includegraphics[width=8cm]{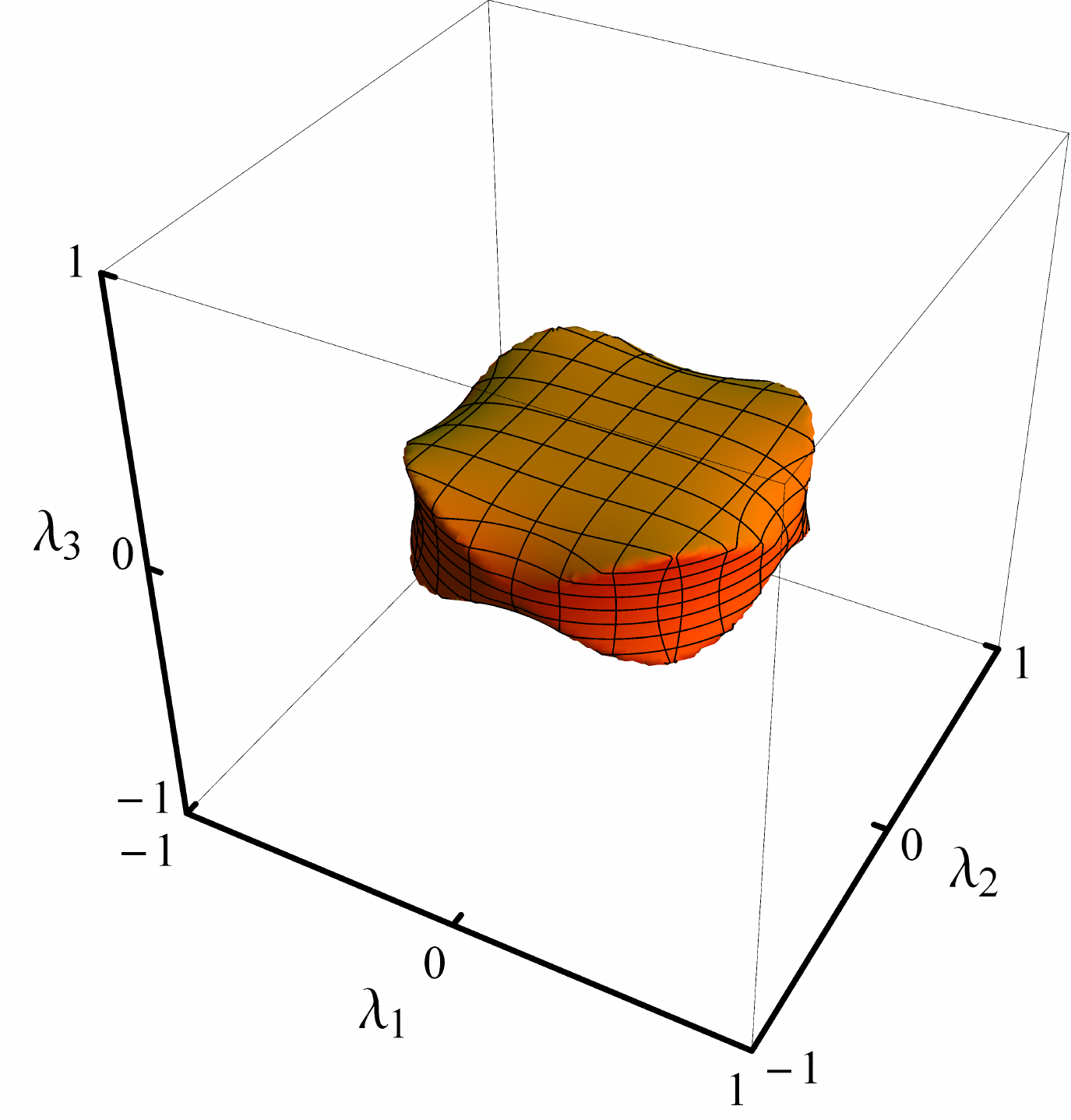}
		\caption{\label{figure-non-unital-numerical} Region of parameters
			$\lambda_1$, $\lambda_2$, $\lambda_3$, where the non-unital map
			$\Phi$ defined by Eq.~\eqref{matrix-form-non-unital} with $t=0.8$
			is 2-tensor-stable positive.}
	\end{figure}
	
	\begin{remark}
		The result of Ref.~\cite{aubrun-szarek-2015} is that the matrices
		$A^{-1}$ and $B^{-1}$ can be chosen positive-definite for a map
		$\Phi$ from the interior of the cone of positive maps. In
		Proposition~\ref{proposition-reduction-to-unitary}, we have
		considered non-degenerate matrices $A^{-1}$ and $B^{-1}$.
	\end{remark}
	
	It is not hard to see that in contrast to the case of unital maps,
	the condition $(\Phi\otimes\Phi)[\ket{\psi_{+}}\bra{\psi_{+}}]
	\geqslant 0$, where
	$\ket{\psi_{+}}=\frac{1}{\sqrt{2}}(\ket{00}+\ket{11})$, is not a
	sufficient condition for positivity of a non-unital map
	$\Phi\otimes\Phi$. In fact, direct calculation of eigenvalues of
	$(\Phi\otimes\Phi)[\ket{\psi_{+}}\bra{\psi_{+}}]$ results in the
	following conditions:
	\begin{equation}
	\label{conditions-necessary-non-unital-entangled} \left\{
	\begin{array}{l}
	1-t^2+\lambda_1^2 - \lambda_2^2 - 2t\lambda_3 - \lambda_3^2 \geqslant 0, \\
	1-t^2-\lambda_1^2 + \lambda_2^2 - 2t\lambda_3 - \lambda_3^2 \geqslant 0, \\
	1-t^2-\sqrt{4t^2 + (\lambda_1^2+\lambda_2^2)^2} + \lambda_3^2 \geqslant 0, \\
	1-t^2+\sqrt{4t^2 + (\lambda_1^2+\lambda_2^2)^2} + \lambda_3^2 \geqslant 0. \\
	\end{array}
	\right.
	\end{equation}
	
	The area of parameters $\lambda_1,\lambda_2,\lambda_3$ satisfying
	inequalities~\eqref{conditions-necessary-non-unital-entangled} for
	a fixed value of parameter $t$ is shown in
	Fig.~\ref{figure-non-unital-entangled}.
	
	However, from Eq.~\eqref{Phi-through-Upsilon-non-unital} and
	Proposition~\ref{proposition-2-tsp} it follows that a map $\Phi$
	from the interior of cone of positive maps is 2-tensor-stable
	positive if and only if its action on the pure state $A^{-1}
	\otimes A^{-1} \ket{\psi_+}$ results in the positive-semidefinite
	operator. Consequently, such a map $\Phi$ is 2-tensor-stable
	positive if and only if
	\begin{equation}
	\widetilde{\lambda}_0^2 \pm \widetilde{\lambda}_{3}^2 \geqslant
	|\widetilde{\lambda}_{1}^2 \pm \widetilde{\lambda}_{2}^2|.
	\end{equation}
	
	The region of parameters $\lambda_1$, $\lambda_2$, $\lambda_3$,
	where the map $\Phi\otimes\Phi$ is positive for a fixed $t$, is
	shown in Fig.~\ref{figure-non-unital-numerical}. Although
	Fig.~\ref{figure-non-unital-numerical} looks like an intersection
	of regions depicted in Figs.~\ref{figure-non-unital-positive} and
	\ref{figure-non-unital-entangled}, it is not.
	
	Finally, one can proceed analogously to find necessary and
	sufficient conditions for 2-tensor-stable positivity of maps
	defined by the matrix
	representation~\eqref{matrix-form-non-unital-general}.
	
	
	\section{\label{section-3-tensor-stable-positive-unital} 3-tensor-stable positive qubit maps}
	
	Let us proceed to higher-order tensor-stable positive maps,
	namely, a unital subclass of 3-tensor-stable positive qubit maps
	$\Phi$. Similarly to the results of
	Sec.~\ref{section-2-tensor-stable-positive-unital}, the map
	$\Phi^{\otimes 3}$ is positive if and only if $\Upsilon^{\otimes
		3}$ is positive, with the diagonal map $\Upsilon$ being
	parameterized by Eq.~\eqref{Upsilon}.
	
	First, we analytically find necessary conditions for positivity of
	the map $\Upsilon^{\otimes 3}$.
	
	\begin{proposition}
		\label{proposition-3-tsp} If the unital qubit map $\Upsilon$ is
		3-tensor-stable positive, then the following 12 inequalities are
		satisfied:
		\begin{eqnarray}
		&& \label{3-tsp-1} 1 - \lambda_i^3 - 3 \lambda_i \lambda_j^2 + 3 \lambda_k^2 \geqslant 0, \\
		&& \label{3-tsp-2} 1 + \lambda_i^3 + 3 \lambda_i \lambda_j^2 + 3
		\lambda_k^2 \geqslant 0,
		\end{eqnarray}
		
		\noindent where $(i,j,k)$ is a permutation of indices $(1,2,3)$,
		i.e. $i,j,k=1,2,3$ and $i \neq j \neq k \neq i$.
	\end{proposition}
	\begin{proof}
		Consider the three-qubit Greenberger?-Horne?-Zeilinger
		state~\cite{greenberger-horne-zeilinger-1989,greenberger-1990}
		\begin{equation}
		\ket{{\rm GHZ}} = \frac{1}{\sqrt{2}} ( \ket{000} + \ket{111} )
		\end{equation}
		
		\noindent written in the basis, in which the map $\Upsilon$ has
		the form \eqref{Upsilon}. Let us define permutations generated by
		the following matrices:
		\begin{eqnarray}
		\label{u-1} && u_1 = \frac{\sigma_2 + \sigma_3}{\sqrt{2}} = \frac{1}{\sqrt{2}} \left(%
		\begin{array}{cc}
		1 & -i \\
		i & -1 \\
		\end{array}%
		\right),\\
		\label{u-2} && u_2 = \frac{\sigma_1 + \sigma_3}{\sqrt{2}} = \frac{1}{\sqrt{2}} \left(%
		\begin{array}{cc}
		1 & 1 \\
		1 & -1 \\
		\end{array}%
		\right), \nonumber\\
		\label{u-3} && u_3 = \frac{\sigma_1 + \sigma_2}{\sqrt{2}} = \frac{1}{\sqrt{2}} \left(%
		\begin{array}{cc}
		0 & 1-i \\
		1+i & 0 \\
		\end{array}%
		\right).
		\end{eqnarray}
		
		\noindent The physical meaning of unitary transformation $u_i
		\cdot u_i^{\dag}$ is the rotation of the Bloch ball such that the
		$i$-th direction becomes inverted and the other two perpendicular
		directions ($j$-th and $k$-th) become interchanged. Denote $U_0 =
		I \otimes I \otimes I$ and $U_i = u_i \otimes u_i \otimes u_i$,
		$i=1,2,3$. We generate the following transformations of the GHZ
		state:
		\begin{equation}
		\varrho_{ij} = U_i U_j \ket{{\rm GHZ}} \bra{{\rm GHZ}} (U_i
		U_j)^{\dag}, \quad i,j=0,1,2,3.
		\end{equation}
		
		If $\Upsilon$ is 3-tensor-stable positive, then
		\begin{equation}
		(\Upsilon\otimes\Upsilon\otimes\Upsilon)[\varrho_{ij}] \geqslant
		0, \quad i,j=0,1,2,3,
		\end{equation}
		
		\noindent which is a number of constraints on parameters
		$\lambda_1,\lambda_2,\lambda_3$. Intersection of these constraints
		results in 12 inequalities \eqref{3-tsp-1}--\eqref{3-tsp-2}.
	\end{proof}
	
	\begin{figure}
		\includegraphics[width=9cm]{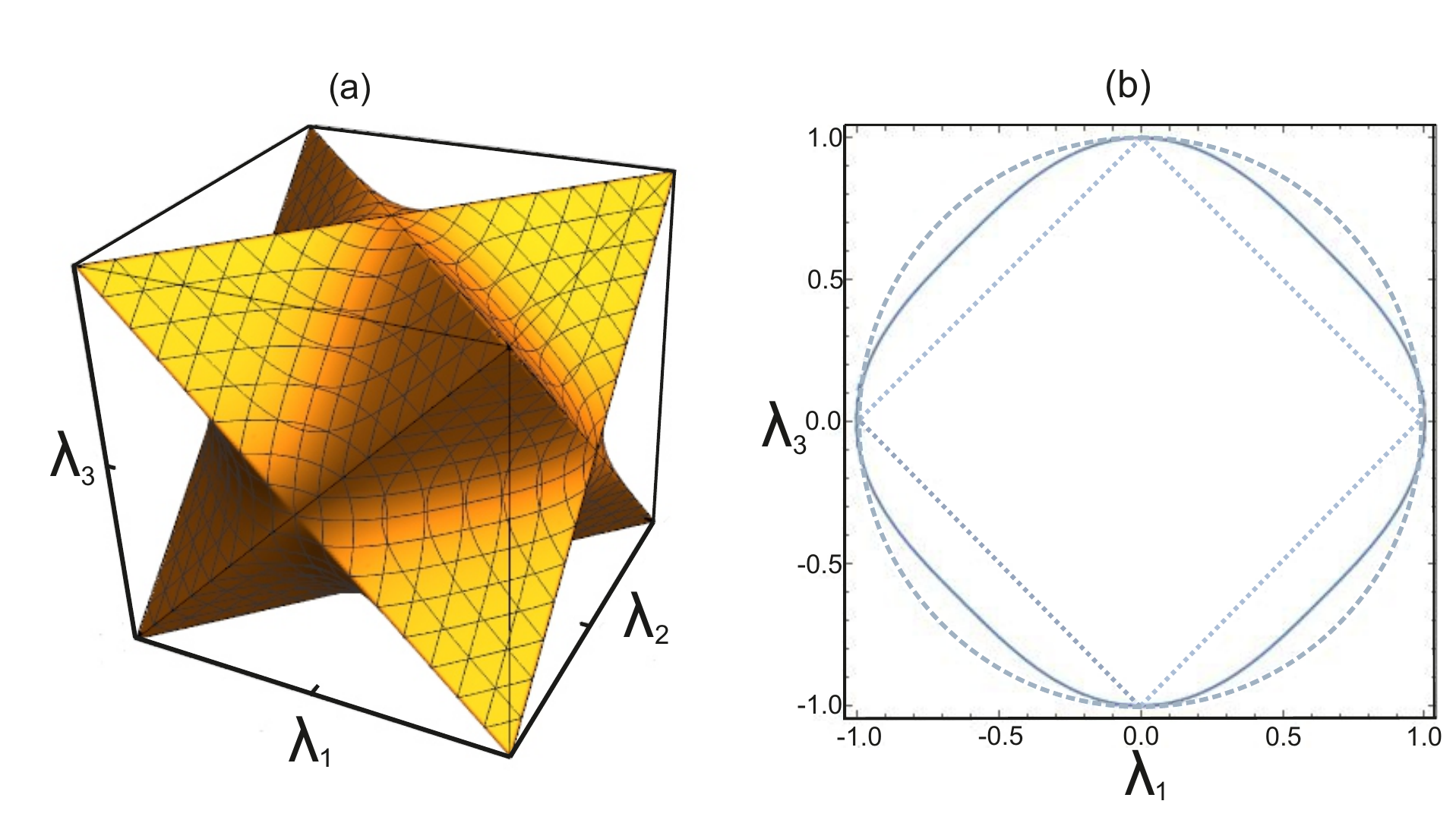}
		\caption{ \label{figure-3-tsp} (a) Region of parameters
			$\lambda_1,\lambda_2,\lambda_3$, where the unital qubit map
			$\Upsilon$ (see Eq.~\eqref{Upsilon}) is 3-tensor-stable positive.
			(b) The cut $\lambda_2=0$: dashed line is a boundary of region,
			where $\Upsilon$ is 2-tensor-stable positive; solid line is a
			boundary of region, where $\Upsilon$ is 3-tensor-stable positive;
			dotted line is a boundary of region, where $\Upsilon$ is
			tensor-stable positive.}
	\end{figure}
	
	In view of complexity of inequalities
	\eqref{3-tsp-1}--\eqref{3-tsp-2}, we have used numerical methods
	to analyze the block-positivity of the Choi operator
	$\Omega_{\Upsilon^{\otimes 3}}$ with respect to the cut $123|123$.
	It turns out that $\Omega_{\Upsilon^{\otimes 3}}$ is block
	positive whenever parameters $\lambda_1,\lambda_2,\lambda_3$
	satisfy \eqref{3-tsp-1}--\eqref{3-tsp-2}. Therefore, there is a
	numerical evidence that Proposition~\ref{proposition-3-tsp}
	provides not only a necessary but also a sufficient condition for
	positivity of the map $\Upsilon^{\otimes 3}$.
	
	The region of parameters $(\lambda_1,\lambda_2,\lambda_3)$
	satisfying \eqref{3-tsp-1}--\eqref{3-tsp-2} is depicted in
	Fig.~\ref{figure-3-tsp}. One can readily see that 3-tensor-stable
	positive maps occupy a subset of 2-tensor-stable positive maps and
	contain the set of tensor-stable positive qubit maps, which is
	known to consist of trivial maps only (completely positive and
	completely co-positive ones)~\cite{muller-hermes-2016}.
	
	Using the relation~\eqref{Phi-through-Upsilon-non-unital} and
	Proposition~\ref{proposition-3-tsp}, the full characterization of
	non-unital 3-tensor-stable positive qubit maps follows
	straightforwardly.
	
	
	\section{\label{section-n-tsp} $n$-tensor-stable positive qubit maps}
	
	A general positive qubit map $\Phi$ is $n$-tensor-stable positive
	if and only if the corresponding Pauli map $\Upsilon$ is
	$n$-tensor stable positive (specified by formula
	\eqref{Phi-through-Upsilon} for unital maps $\Phi$ and by formula
	\eqref{Phi-through-Upsilon-non-unital} for non-unital maps
	$\Phi$).
	
	\begin{proposition}
		Suppose the Pauli qubit map $\Upsilon$ is $n$-tensor-stable
		positive, then
		\begin{eqnarray}
		\label{condition-n-tsp}
		&& | (1 + \lambda_i)^p (1 - \lambda_i)^q + (1 - \lambda_i)^p (1 + \lambda_i)^q | \geqslant \nonumber\\
		&&  \geqslant | (\lambda_j + \lambda_k)^p (\lambda_j - \lambda_k)^q + (\pm 1)^n (\lambda_j - \lambda_k)^p (\lambda_j + \lambda_k)^q |, \nonumber\\
		\end{eqnarray}
		
		\noindent for all permutations $(i,j,k)$ of indices $(1,2,3)$ and
		all $p,q \in {\mathbb Z}_+$ such that $p+q=n$.
	\end{proposition}
	
	\begin{proof}
		
		Consider a generalized GHZ state of $n$ qubits,
		$\ket{\text{GHZ}_n} = \frac{1}{\sqrt{2}}(\ket{0}^{\otimes n} +
		\ket{1}^{\otimes n})$. Note that
		$\ket{\text{GHZ}_n}\bra{\text{GHZ}_n} = \frac{1}{2} [ (\sigma_+
		\sigma_-)^{\otimes n} + \sigma_+^{\otimes n} + \sigma_-^{\otimes
			n} + (\sigma_- \sigma_+ )^{\otimes n}]$, where $\sigma_{\pm} =
		\frac{1}{2}(\sigma_1 \pm i \sigma_2)$. Since
		$\Upsilon[\sigma_{\pm} \sigma_{\mp}] = \sigma_0 \pm \lambda_3
		\sigma_3$ and $\Upsilon[\sigma_{\pm}] = \lambda_1 \sigma_1 \pm i
		\lambda_2 \sigma_2$, the operator
		\begin{eqnarray}
		\!\!\!\!\!\!\!\!\!\! && \Upsilon^{\otimes n} [\ket{\text{GHZ}_n}\bra{\text{GHZ}_n}] = 2^{-(n+1)} \Big[ (\sigma_0 + \lambda_3 \sigma_3)^{\otimes n} \nonumber\\
		\!\!\!\!\!\!\!\!\!\! && + (\lambda_1 \sigma_1 + i \lambda_2 \sigma_2)^{\otimes n} + (\lambda_1 \sigma_1 - i \lambda_2 \sigma_2)^{\otimes n} + (\sigma_0 - \lambda_3 \sigma_3)^{\otimes n} \Big] \nonumber
		\end{eqnarray}
		
		\noindent has so-called $X$-form in the conventional basis. In
		accordance with Sylvester's criterion, such an operator is
		positive semidefinite if and only if $(1 + \lambda_3)^p (1 -
		\lambda_3)^q + (1 - \lambda_3)^p (1 + \lambda_3)^q | \geqslant |
		(\lambda_1 + \lambda_2)^p (\lambda_1 - \lambda_2)^q + (\lambda_1 -
		\lambda_2)^p (\lambda_1 + \lambda_2)^q | $ for all $p,q \in
		{\mathbb Z}_+$, $p+q = n$. Continuing the same line of reasoning
		for states $u_i^{\otimes n}\ket{\text{GHZ}_n}$ and $(u_i
		u_j)^{\otimes n}\ket{\text{GHZ}_n}$, where $u_i$ and $u_j$ are
		either identity operators or have the form
		\eqref{u-1}--\eqref{u-3}, we get permutations of
		$(\lambda_1,\lambda_2,\lambda_3)$ accompanied with the appropriate
		sign changes. All the obtained inequalities are summarized in
		Eq.~\eqref{condition-n-tsp}.
	\end{proof}
	
	\begin{corollary}
		\label{corollary-n-powers} The Pauli qubit map $\Upsilon$ satisfies Eq. \eqref{condition-n-tsp} for $n\geqslant 2$ if
		\begin{equation}
		\label{lambda-degrees}
		\sum_{i=1}^3 \lambda_i^{\frac{n}{n-1}} \leqslant 1.
		\end{equation}
	\end{corollary}
	
	\begin{proof}
		In the space of parameters $(\lambda_1,\lambda_2,\lambda_3)$, the
		geometrical figure \eqref{condition-n-tsp} comprises the figure
		\eqref{lambda-degrees}. The surfaces of two figures touch at
		points satisfying $\lambda_i = \lambda_j =0$, $|\lambda_k| = 1$,
		or $|\lambda_i| = |\lambda_j| = 2^{-\frac{n-1}{n}}$, $\lambda_k =
		0$.
	\end{proof}
	
	The statement of Corollary~\ref{corollary-n-powers} is valid for
	all $n=2,3,\ldots$ and stimulates the discussion of recurrence
	relation between $n$- and $(n+1)$-tensor-stable positive maps. In
	fact, suppose the map $\Phi$ is $n$-tensor-stable positive. Is it
	possible to modify $\Phi$ and construct a map $\tilde{\Phi}$,
	which is $(n+1)$-tensor-stable positive? The following proposition
	provides an affirmative answer to this question.
	
	\begin{proposition}
		\label{proposition-n-tsp-sufficient} Suppose $\Phi$ is
		$n$-tensor-stable positive and $\Phi_{\rm EB}$ is entanglement
		breaking, then the map
		\begin{equation}
		\tilde{\Phi} = \mu \Phi + (1-\mu) \Phi_{\rm EB}
		\end{equation}
		
		\noindent is $(n+1)$-tensor-stable positive whenever
		\begin{equation}
		\label{mu-condition} \frac{\mu}{1-\mu} \leqslant
		\sqrt[n+1]{\frac{m_{\Phi_{\rm EB}}}{|m_{\rm \Phi}|}}, \quad
		m_{\Phi} = \!\!\!\!\! \underset{\varrho\in({\cal B}({\cal
				H}^{\otimes (n+1)}))^+, \ {\rm tr}[\varrho] = 1}{\rm
			min.~eigenvalue} \!\!\!\!\! \ \Phi^{\otimes (n+1)} [\varrho].
		\end{equation}
	\end{proposition}
	
	\begin{proof}
		Expanding $\tilde{\Phi}^{\otimes (n+1)}$, we notice that the maps
		$\Phi_{\rm EB}\otimes\Phi^{\otimes n}$, $\Phi_{\rm EB}^{\otimes 2}
		\otimes \Phi^{\otimes(n-1)}$, $\ldots$, $\Phi_{\rm EB}^{\otimes
			n}\otimes \Phi$ are all positive by
		Proposition~\ref{proposition-eb-otimes-positive} as $\Phi^{\otimes
			n}$ is positive by the statement. Hence, if the map
		\begin{equation}
		\label{eb-n-plus-phi-n} \mu^{n+1} \Phi^{\otimes(n+1)} +
		(1-\mu)^{n+1}\Phi_{\rm EB}^{\otimes (n+1)}
		\end{equation}
		
		\noindent is positive, then the map $\tilde{\Phi}^{\otimes (n+1)}$
		is positive too. On the other hand, the map
		\eqref{eb-n-plus-phi-n} is positive whenever the minimal output
		eigenvalue is non-negative, which results in formula
		\eqref{mu-condition} and concludes the proof.
	\end{proof}
	
	Applying Proposition~\ref{proposition-n-tsp-sufficient} to the
	Pauli channels $\Upsilon$, we get a recurrent sufficient condition
	for $(n+1)$-tensor-stable positivity.
	
	\begin{proposition}
		\label{proposition-n-tsp-Pauli} Let the Pauli map $\Upsilon$ with
		parameters $\lambda_1,\lambda_2,\lambda_3$ be $n$-tensor-stable
		positive and $|\lambda_1| + |\lambda_2| + |\lambda_3| \geqslant
		1$, then the map $\tilde{\Upsilon}$ with parameters
		\begin{eqnarray}
		\label{lambda-tilde-parameters} & \tilde{\lambda}_i & =
		\frac{(|\lambda_1| + |\lambda_2| +
			|\lambda_3|)^{-1} + x}{1+x} \, \lambda_i, \\
		0 \leqslant & x & \leqslant \frac{1}{2} \left(1  -
		\frac{\max\limits_{k=1,2,3} |\lambda_k| }{|\lambda_1| +
			|\lambda_2| + |\lambda_3|} \right) \sqrt[n+1]{ \frac{2
			}{\max\limits_{k=1,2,3} |\lambda_k|} }, \nonumber
		\end{eqnarray}
		
		\noindent is $(n+1)$-tensor-stable positive.
	\end{proposition}
	
	\begin{proof}
		We use Proposition~\ref{proposition-n-tsp-sufficient}, where the
		map $\Phi_{\rm EB}$ has parameters $\lambda_i^{\rm EB} = \lambda_i
		/ (|\lambda_1| + |\lambda_2| + |\lambda_3|)$. Then $m_{\Phi_{\rm
				EB}} \geqslant 2^{-(n+1)} [ 1 - \max_{k} |\lambda_k^{\rm EB}|
		]^{n+1}$ and $m_{\Phi} \geqslant - \frac{1}{2} \max_{k}
		|\lambda_k|$ as $\Phi^{\otimes n}$ is positive and
		trace-preserving. Substituting the obtained values in
		Eq.~\eqref{mu-condition} and using the explicit form of $\Phi_{\rm
			EB}$, we get parameters \eqref{lambda-tilde-parameters}.
	\end{proof}
	
	\begin{example}
		Consider a family of the Pauli maps $\Upsilon$ with $\lambda_1 =
		\lambda_3 = t$ and $\lambda_2 = 0$ (see Fig.~\ref{figure-3-tsp}b).
		
		$\Upsilon$ is positive ($n=1$) if $|t|\leqslant 1$. By
		Proposition~\ref{proposition-n-tsp-Pauli}, $\tilde{\Upsilon}$ is
		2-tensor-stable positive if $|t| \leqslant 0.63$, which is in
		agreement with the exact result $|t| \leqslant \frac{1}{\sqrt{2}}
		\approx 0.71$ (Proposition~\ref{proposition-2-tsp}).
		
		Let $\Upsilon$ be 2-tensor-stable positive, i.e. $|t| \leqslant
		\frac{1}{\sqrt{2}}$, then
		Proposition~\ref{proposition-n-tsp-Pauli} implies that
		$\tilde{\Upsilon}$ is 3-tensor-stable positive if $|t| \leqslant
		0.55$, which is in agreement with the result of
		Section~\ref{section-3-tensor-stable-positive-unital}, $|t|
		\leqslant 2^{-2/3} \approx 0.63$.
		
		If $\Upsilon$ is 3-tensor-stable positive, i.e. $|t| \leqslant
		2^{-2/3}$, then Proposition~\ref{proposition-n-tsp-Pauli} implies
		that $\tilde{\Upsilon}$ is 4-tensor-stable positive if $|t|
		\leqslant 0.532$.
	\end{example}
	
	
	\section{\label{section-witness} Witnessing entanglement}
	
	Positive maps are often used to detect different types of
	entanglement~\cite{peres-1996,horodecki-1996,terhal-2001,chen-2004,breuer-2006,hall-2006,chruscinski-kossakowski-2007,ha-2011,chruscinski-sarbicki-2014,collins-2016,huber-2014,lancien-2015}.
	In this section, we find particular applications of
	$n$-tensor-stable positive maps in partial characterization of the
	entanglement structure.
	
	A general density operator $\varrho$ of $N$ qubits adopts the
	resolution
	\begin{equation}
	\label{rho-resolution} \varrho = \sum_{j} p_{jk} \varrho_j^{(1)}
	\otimes \cdots \otimes \varrho_j^{(k)}, \quad p_{jk} \geqslant 0,
	\end{equation}
	
	\noindent where $N$ qubits are divided into $k$ parts. For each
	fixed resolution of the state $\varrho$ define the maximal number
	of qubits in the parts, $\max_{m=1,\ldots,k} \# \varrho_j^{(m)}$.
	The resource intensiveness~\cite{filippov-melnikov-ziman-2013}
	(entanglement depth~\cite{sorensen-2001},
	producibility~\cite{guhne-2005}) of the quantum state $\varrho$ is
	defined through
	\begin{equation}
	R_{\rm ent} [\varrho] = \min_{\varrho = \sum_{j} p_{jk}
		\varrho_j^{(1)} \otimes \cdots \otimes \varrho_j^{(k)}}
	\max_{m=1,\ldots,k} \# \varrho_j^{(m)}
	\end{equation}
	
	\noindent and specifies the minimal physical resources needed to
	create such a state, namely, the minimal number of qubits to be
	entangled. The state $\varrho$ is called fully separable if
	$R_{\rm ent} = 1$ and genuinely entangled if $R_{\rm ent} = N$.
	
	The following result enables one to detect the entanglement depth
	via $n$-tensor-stable positive maps.
	
	\begin{proposition}
		\label{proposition-r-entangled} Let $\varrho$ be an $N$-qubit
		state. Suppose $\Phi$ is an $n$-tensor-stable positive qubit map
		and $\Phi^{\otimes N}[\varrho] \ngeqslant 0$ (contains negative
		eigenvalues), then $R_{\rm ent} [\varrho] \geqslant n+1$.
	\end{proposition}
	
	\begin{proof}
		Suppose $R_{\rm ent} [\varrho] \leqslant n$, then there exists a
		resolution \eqref{rho-resolution} such that each state
		$\varrho_j^{(m)}$ comprises at most $n$ qubits. Therefore,
		$\Phi^{\otimes \# \varrho_j^{(m)}} [\varrho_j^{(m)}] \geqslant 0$
		in view of the nested structure of $k$-tensor-stable positive
		maps. Thus, $\Phi^{\otimes N}[\varrho] \geqslant 0$, which leads
		to a contradiction with the statement of proposition. Hence,
		$R_{\rm ent} [\varrho] \geqslant n+1$.
	\end{proof}
	
	In what follows, we illustrate the use of
	Proposition~\ref{proposition-r-entangled} for detecting particular
	forms of multipartite entanglement.
	
	\begin{example}
		A depolarized GHZ state of three qubits
		\begin{equation}
		\varrho_q^{\rm GHZ} = q \ket{{\rm GHZ}}\bra{{\rm GHZ}} + (1-q)
		\frac{1}{8} I, \quad 0 \leqslant q \leqslant 1,
		\end{equation}
		
		\noindent is not fully separable if $q \geqslant 0.26$ as there
		exists a positive Pauli map $\Upsilon$ (with $|\lambda_i|
		\leqslant 1$) such that $\Upsilon^{\otimes 3} [\varrho_q^{\rm
			GHZ}] \ngeqslant 0$. Also, $\varrho_q^{\rm GHZ}$ is genuinely
		entangled if $q \geqslant 0.71$ as there exists a 2-tensor-stable
		positive Pauli map $\Upsilon$ (with parameters
		\eqref{hyperboloids}) such that $\Upsilon^{\otimes 3}
		[\varrho_q^{\rm GHZ}] \ngeqslant 0$.
	\end{example}
	
	\begin{example}
		Consider a depolarized W state of three qubits
		\begin{equation}
		\varrho_q^{\rm W} = q \ket{{\rm W}}\bra{{\rm W}} + (1-q)
		\frac{1}{8} I, \quad 0 \leqslant q \leqslant 1,
		\end{equation}
		
		\noindent where $\ket{W} = \frac{1}{\sqrt{3}}(\ket{100} +
		\ket{010} + \ket{001})$. The state $\varrho_q^{\rm W}$ is not
		fully separable if $q \geqslant 0.31$ as there exists a positive
		Pauli map $\Upsilon$ (with $|\lambda_i| \leqslant 1$) such that
		$\Upsilon^{\otimes 3} [\varrho_q^{\rm W}] \ngeqslant 0$.
		Analogously, $\varrho_q^{\rm W}$ is genuinely entangled if $q
		\geqslant 0.86$ as there exists a 2-tensor-stable positive Pauli
		map $\Upsilon$ (with parameters \eqref{hyperboloids}) such that
		$\Upsilon^{\otimes 3} [\varrho_q^{\rm W}] \ngeqslant 0$.
	\end{example}

	
	\section{\label{section-conclusions} Conclusions}
	
	We have addressed the problem of positivity of tensor products
	$\bigotimes_{i=1}^n \Phi_i = \Phi_1 \otimes \Phi_2 \otimes \ldots
	\Phi_n$ of linear maps $\Phi_i$. In addition to the apparent
	implications $\{$All $\Phi_i$ are completely positive$\} \vee
	\{$All $\Phi_i$ are completely co-positive$\} \Rightarrow \{
	\bigotimes_{i=1}^n \Phi_i$ is positive$\} \Rightarrow \{$All
	$\Phi_i$ are positive$\}$, we have managed to find non-trivial
	sufficient conditions for positivity of $\Phi_1 \otimes \Phi_2$,
	in particular, for unital qubit maps $\Phi_1$ and $\Phi_2$.
	
	2- and 3-tensor-stable positive qubit maps are fully
	characterized. Namely, the explicit criteria for unital maps are
	found (Eq.~\eqref{hyperboloids} and
	Eqs.~\eqref{3-tsp-1}--\eqref{3-tsp-2}, respectively), and the
	analysis of non-unital maps is reduced to the case of unital ones.
	
	Basing on the examples of decomposable positive maps
	$\Phi^{\otimes 2}$, we have conjectured that all positive
	two-qubit maps $\Phi^{\otimes 2}$ are decomposable.
	
	For $n$-tensor-stable positive qubit maps we have found necessary
	and (separately) sufficient conditions. The first necessary
	condition involves algebraic inequalities on parameters
	$\lambda_1$, $\lambda_2$, $\lambda_3$ of degree $n$. Another condition has a concise form and clearly shows the
	nested structure of maps. The sufficient conditions have a
	recurrent form and enable one to find $(n+1)$-tensor-stable
	positive maps once $n$-tensor-stable positive maps are known.
	Entanglement breaking channels play a vital role in the derivation
	of those recurrent formulas. Due to the
	relation~\eqref{Phi-through-Upsilon-non-unital}, the results
	obtained for unital maps can be readily transferred to non-unital
	maps.
	
	Finally, we have discussed the application of positive maps with
	tensor structure to characterization of multipartite entanglement.
	A criterion for quantifying the entanglement depth (resource
	intensiveness, producibility) via $n$-tensor-stable positive maps
	is found and illustrated by a number of examples, which detect the
	genuine entanglement and the absence of full separability in
	depolarized GHZ and W states.
	
	
	\bigskip
	
	\begin{acknowledgments}
		We gratefully thank David Reeb for fruitful comments and for
		bringing Refs.~\cite{gurvits-2004,aubrun-szarek-2015} to our
		attention, which enabled us to extend the exact results obtained
		for unital maps to the case of non-unital maps. The study is
		supported by Russian Science Foundation under project No.
		16-11-00084 and performed in Moscow Institute of Physics and
		Technology.
	\end{acknowledgments}
	

\end{document}